\documentclass[aps,prl,twocolumn,preprintnumbers,superscriptaddress,nofootinbib]{revtex4-2}

\usepackage[utf8]{inputenc}

\usepackage{amsmath,amssymb,mathtools,amsthm,mathrsfs,physics,amsfonts,dsfont}
\usepackage[pass]{geometry}
\usepackage{pdfrender} %

\usepackage{bibunits}

\usepackage{svg}
\usepackage{graphicx}
\usepackage{tikz}
\usetikzlibrary{quantikz2}

\usepackage{algorithm,algpseudocode}

\usepackage{tabularx,siunitx,booktabs}

\usepackage{xcolor}
\definecolor{lavender}{rgb}{0.64,0.47,0.87}
\definecolor{darklavender}{rgb}{0.48,0.35,0.65}
\definecolor{revisionblue}{RGB}{0,0,0}

\usepackage{ytableau}

\usepackage{caption,subcaption}
\usepackage[hypertexnames=false]{hyperref}
\hypersetup{
    colorlinks=true,
    linkcolor=darklavender,
    citecolor=darklavender,
    urlcolor=darklavender,
}

\usepackage{enumitem,comment,adjustbox,extarrows}
\usepackage{float} %
\usepackage[normalem]{ulem}

\usepackage{etoolbox}
\usepackage{footmisc}
\usepackage{aliascnt}

\makeatletter
\renewcommand\listoffigures{\@starttoc{lof}}
\renewcommand\listoftables{\@starttoc{lot}}

\makeatother

\newtheorem{theorem}{Theorem}[section]
\theoremstyle{definition}
\newtheorem{definition}{Definition}
\newaliascnt{lemma}{theorem}
\newtheorem{lemma}[lemma]{Lemma}
\aliascntresetthe{lemma}
\newaliascnt{corollary}{theorem}
\newtheorem{corollary}[corollary]{Corollary}
\aliascntresetthe{corollary}
\newaliascnt{notation}{theorem}

\aliascntresetthe{notation}
\newaliascnt{proposition}{theorem}

\aliascntresetthe{proposition}
\newaliascnt{conjecture}{theorem}

\aliascntresetthe{conjecture}
\newaliascnt{assumption}{theorem}

\aliascntresetthe{assumption}
\newaliascnt{observation}{theorem}

\aliascntresetthe{observation}
\newaliascnt{fact}{theorem}

\aliascntresetthe{fact}
\newaliascnt{remark}{theorem}

\aliascntresetthe{remark}
\newaliascnt{claim}{theorem}

\aliascntresetthe{claim}
\newaliascnt{example}{theorem}

\aliascntresetthe{example}
\newaliascnt{problem}{theorem}

\aliascntresetthe{problem}
\newaliascnt{open}{theorem}

\aliascntresetthe{open}
\newaliascnt{hypothesis}{theorem}

\aliascntresetthe{hypothesis}
\newaliascnt{question}{theorem}

\aliascntresetthe{question}

\usepackage[nameinlink]{cleveref}

\definecolor{pink}{rgb}{0.72,0.18,0.95}
\definecolor{celeste}{rgb}{0.09, 0.81, 0.9}
\definecolor{emerald}{rgb}{0, 0.677778, 0.555556}

\crefname{lemma}{Lemma}{Lemmas}
\crefname{proposition}{Proposition}{Propositions}
\crefname{definition}{Definition}{Definitions}
\crefname{theorem}{Theorem}{Theorems}
\crefname{conjecture}{Conjecture}{Conjectures}
\crefname{corollary}{Corollary}{Corollaries}
\crefname{claim}{Claim}{Claims}
\crefname{section}{Section}{Sections}
\crefname{appendix}{Appendix}{Appendices}
\crefname{figure}{Fig.}{Figs.}
\crefname{equation}{Eq.}{Eqs.}
\crefname{table}{Table}{Tables}
\theoremstyle{definition}

\def\multiset#1#2{\ensuremath{\left(\kern-.3em\left(\genfrac{}{}{0pt}{}{#1}{#2}\right)\kern-.3em\right)}}

\newcommand{\beginsupplement}{
  \clearpage
  \onecolumngrid
  \newgeometry{margin=1in}
  \hypersetup{pageanchor=false}
  \setcounter{secnumdepth}{3}
  \setcounter{page}{1}
  \setcounter{section}{0}
\renewcommand{\thesection}{S\arabic{section}}
\setcounter{theorem}{0}
\renewcommand{\thetheorem}{S\arabic{section}.\arabic{theorem}}
\renewcommand{\theHtheorem}{S\arabic{section}.\arabic{theorem}}
  \setcounter{equation}{0}
  \renewcommand{\theequation}{S\arabic{equation}}
  \renewcommand{\theHequation}{S.\arabic{equation}}
  \setcounter{table}{0}
  \renewcommand{\thetable}{S\arabic{table}}
  \renewcommand{\theHtable}{S.\arabic{table}}
  \setcounter{figure}{0}
  \renewcommand{\thefigure}{S\arabic{figure}}
  \renewcommand{\theHfigure}{S.\arabic{figure}}
  \renewcommand{\bibnumfmt}[1]{[S##1]}
\renewcommand{\citenumfont}[1]{S##1}}

\usepackage{mleftright}
\usepackage{xparse}

\ExplSyntaxOn
\cs_set_protected:Npn \coh_asymp_print:nnnn #1#2#3#4
  {
    #1
    \tl_if_blank:nF {#3} { \sp{(#3)} }
    \tl_if_blank:nF {#2} { \sb{#2} }
    \mleft( #4 \mright)
  }

\DeclareDocumentCommand{\bigO}{O{} O{} m}
  {
    \ensuremath{\coh_asymp_print:nnnn {O} {#1} {#2} {#3}}
  }
\DeclareDocumentCommand{\lilo}{O{} O{} m}
  {
    \ensuremath{\coh_asymp_print:nnnn {o} {#1} {#2} {#3}}
  }
\DeclareDocumentCommand{\bigOm}{O{} O{} m}
  {
    \ensuremath{\coh_asymp_print:nnnn {\Omega} {#1} {#2} {#3}}
  }
\DeclareDocumentCommand{\lilom}{O{} O{} m}
  {
    \ensuremath{\coh_asymp_print:nnnn {\omega} {#1} {#2} {#3}}
  }
\DeclareDocumentCommand{\bigTheta}{O{} O{} m}
  {
    \ensuremath{\coh_asymp_print:nnnn {\Theta} {#1} {#2} {#3}}
  }
\ExplSyntaxOff

\newcommand{\yd}[1]{{\mleft[#1\mright]}}
\newcommand{\yt}[1]{{\mleft(#1\mright)}}
\newcommand{\wt}[1]{{\mleft\{#1\mright\}}}

\def\multiset#1#2{\ensuremath{\mleft(\kern-.3em\mleft(\genfrac{}{}{0pt}{}{#1}{#2}\mright)\kern-.3em\mright)}}

\newcommand{\triple}[3]{%
  {\begin{array}{@{}c@{}c@{}} & #1\,#2\\ & #3 \end{array}}%
}

\newcommand{\res}[4]{%
  #1^{\scalebox{0.6}{$\triple{#2}{#3}{#4}$}}%
}

\newcommand{\intertwiner}[5]{#1^{#2\xrightarrow{#5}#3\,#4}}
\newcommand{\intertwineradj}[5]{#1^{#3\,#4\xrightarrow{#5}#2}}

\NewDocumentCommand{\extChannel}{m m m m o}{%
  #1^{#2 \xrightarrow{\IfValueTF{#5}{#4,#5}{#4}} \,#3}%
}

\newcommand{\rpoch}[2]{\mleft(#1\mright)^{\overline{#2}}} %

\NewDocumentCommand{\schur}{o m o}{%
  \ensuremath{s^{#2}\IfValueT{#1}{_{#1}}\IfValueT{#3}{\!\mleft(#3\mright)}}%
}

\usepackage{xcolor}
\definecolor{lavender}{rgb}{0.64,0.47,0.87}
\definecolor{darklavender}{rgb}{0.48,0.35,0.65}
\definecolor{blue}{rgb}{0.12156862745098039,0.4666666666666667,0.7058823529411764}
\definecolor{paleorange}{rgb}{0.8431372549019608,0.5568627450980392,0.12549019607843137}
\definecolor{defgreen}{rgb}{0.28627450980392155,0.5490196078431373,0.2549019607843137}
\definecolor{jam}{rgb}{0.7686274509803921, 0.34901960784313724, 0.6313725490196078}

\begin{document}

\title{An Exponential Sample-Complexity Advantage for Coherent Quantum Inference}

\author{Zhaoyi Li}
\thanks{These authors contributed equally.\\Emails: \texttt{ladmon@mit.edu}, \texttt{edmt@math.ku.dk}.}
\affiliation{Department of Physics, Massachusetts Institute of Technology, Cambridge MA 02139, USA}
\affiliation{The NSF AI Institute for Artificial Intelligence and Fundamental Interactions}

\author{Elias Theil}
\thanks{These authors contributed equally.\\Emails: \texttt{ladmon@mit.edu}, \texttt{edmt@math.ku.dk}.}
\affiliation{Department of Mathematical Sciences, University of Copenhagen, 2100 Copenhagen, Denmark}

\author{Aram W. Harrow}
\affiliation{Department of Physics, Massachusetts Institute of Technology, Cambridge MA 02139, USA}

\author{Isaac Chuang}
\affiliation{Department of Physics, Massachusetts Institute of Technology, Cambridge MA 02139, USA}

\date{\today}

\begin{abstract}
Standard quantum inference converts quantum data into classical outputs. We study an alternative inference setting in which the desired output is quantum. Such settings include quantum purity amplification (QPA), random purification and approximate cloning, and density matrix exponentiation. We show that such protocols can achieve exponentially lower sample complexity than incoherent, measurement-mediated protocols, establishing, to the best of our knowledge, the first nontrivial separations between dimension-independent coherent and dimension-dependent incoherent sample complexity. For QPA with principal-eigenstate targets and $d$-dimensional inputs, at fixed output number and spectral gap, coherent processing achieves error $\varepsilon$ using $O(1/\varepsilon)$ copies, whereas any incoherent protocol requires $\Omega(d/\varepsilon)$ copies. Together, these separations motivate a unified picture of coherent quantum inference, in which an entanglement-breaking limit identifies the optimal incoherent counterpart.
\end{abstract}

\makeatletter
\let\frontmatter@footnote@produce\frontmatter@footnote@produce@footnote
\makeatother

\maketitle
\hypertarget{frontmatter.1}{}

\begin{bibunit}[apsrev4-2]

Quantum inference extracts a desired output from limited access to an unknown state, and its central cost is the sample complexity, the number of copies required. Conventionally, this output is classical, such as an expectation value, a classification label, or a tomographic description. Many tasks of interest, however, call for an output that is itself quantum, a state or a channel rather than a classical estimate. Practically, quantum outputs are needed directly as resource states in gate teleportation~\cite{DGHM+25}, as quantum data for subsequent coherent processing~\cite{SS06}, or as inputs for query models~\cite{TWZ25}. Foundationally, as in the theory of unspeakable quantum information~\cite{peres2002unspeakable}, quantities such as direction, phase, time, and helicity are not merely classical labels, and accessing them requires physical reference systems that are themselves quantum. These tasks therefore demand inference that produces a quantum output, which no classical estimate can provide.

This perspective motivates processing quantum information coherently at its true quantum limit. Coherence is costly to preserve in practice, but it is already known to be essential across quantum information. Quantum channels must preserve coherence to transmit quantum information~\cite{HSR03}, quantum simulation requires coherent evolution to reproduce quantum dynamics with the native scaling of the underlying process~\cite{V03}, and quantum algorithms exploit coherence as a computational resource to achieve speedups~\cite{GSLW19,Martyn21}. One might therefore expect coherence to benefit inference tasks as well, yielding a sample-complexity advantage over incoherent protocols, in which the inference inputs are measured and only the resulting classical information is available for subsequent processing.

At first glance, incoherent protocols may seem inherently limited.
Quantum-to-classical conversion based on tomography~\cite{HHJW+16,OW16}, or more broadly on quantum statistical inference~\cite{BGJ03,KKL25} and learning~\cite{AD17,AAKS21}, suffers from dimension-dependent sample complexity, i.e., the ``curse of dimensionality''.
This suggests that incoherent protocols may require sample complexities that scale polynomially in the Hilbert-space dimension $d$, and therefore exponentially in the local system size, namely the number of qubits.

Establishing such a coherent--incoherent separation rigorously is difficult.
On the incoherent side, classical shadows~\cite{HKP20} can reduce the sample complexity of predicting selected properties, bypassing full tomography, while adaptivity and entangled measurements can further improve incoherent estimation.
On the coherent side, admissible protocols encompass all quantum channels. Both sides are therefore enormous optimization classes, and tight sample-complexity bounds are largely unavailable when the desired output is itself quantum. It is a priori unclear that any dimension-scaling separation should exist.

Yet a separation already appears for a trivial example. Consider the task shown in~\cref{fig:identity_coherent_vs_incoherent}: given $n$ copies of an unknown \(d\)-dimensional pure state \(\ket{\psi}\), the goal is to output a single copy of it. A coherent protocol accomplishes this by forwarding the input, so a single copy suffices. By contrast, the performance of an incoherent protocol is limited by pure state tomography; achieving constant output accuracy requires \(n=\bigTheta{d}\) copies~\cite{H13}.
\begin{figure}
    \centering
\includegraphics[width=\columnwidth]{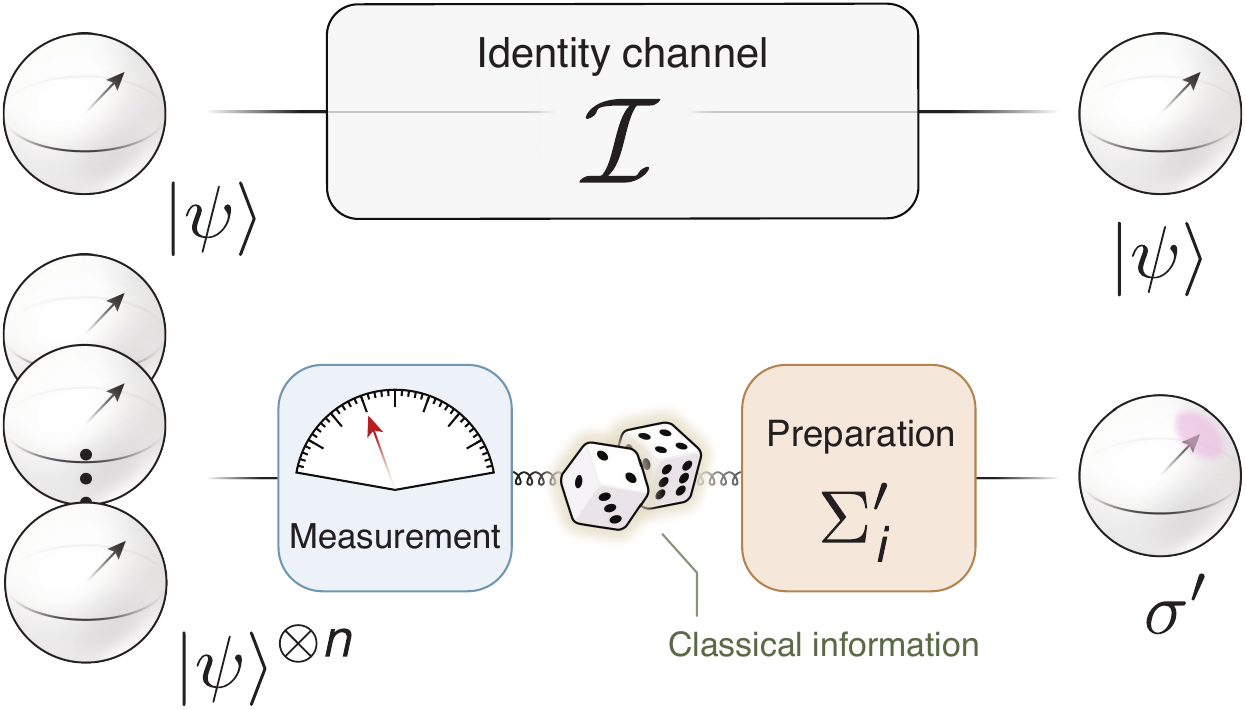}
    \caption{
    Comparison between coherent and incoherent strategies for reproducing an unknown pure state $\rho=\ketbra{\psi}{\psi}$.
    }
\label{fig:identity_coherent_vs_incoherent}
\end{figure}

While the identity task already hints at a coherent advantage, it remains unclear whether tasks that genuinely process the input exhibit a similar separation. Clock distillation and quantum-enhanced learning of qubit rotations provide such examples~\cite{KM25,MC19}, but both tasks have fixed local dimension and therefore do not address separations that grow with local system size.
For larger systems, quantum purity amplification (QPA)~\cite{LFIC24} and density matrix exponentiation (DME)~\cite{LMR14} have been studied, but previous work falls short of a rigorous separation. Existing QPA bounds are either asymptotic without dimension-uniform control~\cite{LFIC24} or too loose to recover the leading constants of the asymptotic regime~\cite{DGHM+25,GLLP+25}; for both QPA and DME, the optimal measurement-mediated risk is not sufficiently understood to quantify the separation.
This leaves a question open: for nontrivial quantum-output inference tasks, can coherent processing provide a provable sample-complexity advantage over any incoherent protocol?

We answer this question by developing a theory of coherent quantum inference (CQI) instantiated by three representative tasks: random purification (RP)~\cite{TWZ25,PSTW25,WW25,GML26}, QPA~\cite{LFIC24}, and DME~\cite{LMR14}. Each yields a coherent--incoherent separation.
For RP, a coherent protocol achieves error $\varepsilon$ for $\ell$ additional clones with $n=\bigO{\sqrt{\ell dr/\varepsilon}}$ copies, while any incoherent protocol requires $n=\bigOm{d/\varepsilon}$.
Our construction simultaneously generalizes the random-purification channel~\cite{TWZ25,PSTW25,WW25,GML26} to the $m>n$ regime and the optimal pure-state $n\to m$ cloner~\cite{W98,KW99,CIGA05,YC13} to inputs of rank at most $r$.
For QPA, we prove a separation by upper bounding the sample complexity of the optimal coherent QPA protocol by \(O(1/\varepsilon)\) and lower bounding the incoherent counterpart, through the entanglement-breaking (EB) limit, by \(\Omega(d/\varepsilon)\).
This is enabled by treating the general \(n\to m\), arbitrary-eigenstate setting, with sharp control of both leading prefactors and admissible regimes, substantially extending previously studied single-output, principal-eigenstate results~\cite{LFIC24,DGHM+25}. Full proofs appear in the companion paper~\cite{QPA26}.
Finally, for DME, the coherent protocol~\cite{LMR14} achieves error \(\varepsilon\) with $n=\bigO{T^2/\varepsilon}$ copies independent of \(d\). We prove that, for sufficiently small $\varepsilon$, with a threshold depending only on $T$, any partially coherent implementation requires $n=\bigOm[][\varepsilon]{\sin^2\left(T/2\right)d/\varepsilon}$, showing a constant-versus-linear-in-$d$ separation.
These separations concern sample complexity rather than implementation complexity. Nevertheless, known implementations of both QPA and DME have gate complexities with only polylogarithmic dependence on the Hilbert-space dimension~\cite{QPA26,kimmel2017hamiltonian}.
Although these tasks differ, we unify them through a common structure: the EB limit sends coherent protocols to incoherent ones, while a complementary prescription promotes classical-output inference tasks to coherent ones.

\emph{Example 1: Random purification and approximate cloning}.\,---
Random purification transforms \(n\) copies of a state \(\rho\) satisfying \(\operatorname{rank}(\rho)\leq r\) into a random mixture of outputs, each consisting of \(m\) copies of some purification \(\ket{\Psi}\in\mathbb{C}^d\otimes\mathbb{C}^r\) satisfying \(\Tr_{\mathbb{C}^r}\Psi=\rho\). Intuitively, the goal is to approximate the transformation \(\rho^{\otimes n}\leadsto\Psi^{\otimes m}\), as illustrated in~\cref{fig:ci_examples}. This task arises naturally in tomography, cryptography, and simulation~\cite{TWZ25,GML26,PSTW25,WW25}.

\begin{figure}[t]
  \centering
\includegraphics[width=0.99\columnwidth]{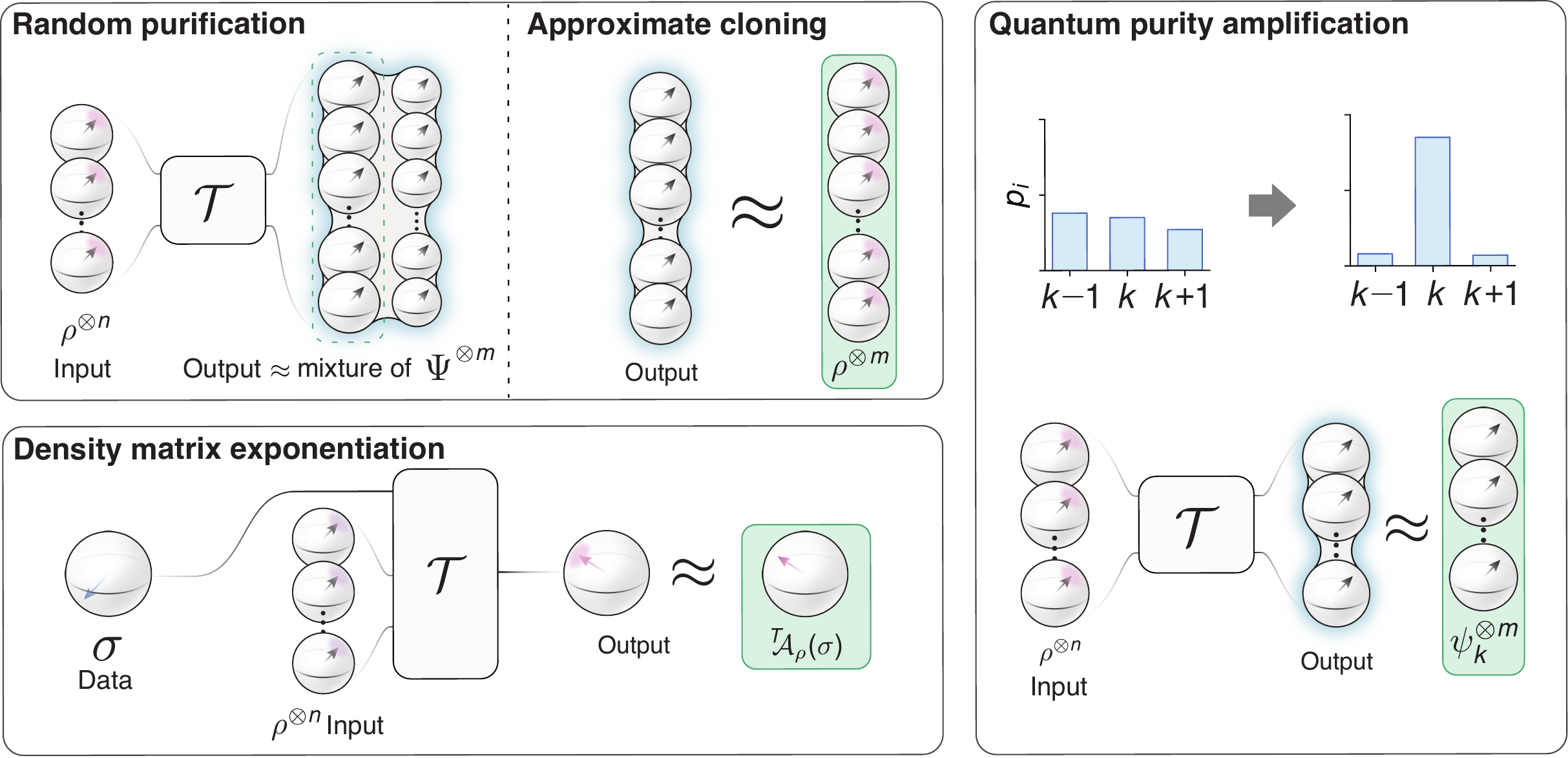}
  \caption{The three coherent-inference examples studied here: random purification and approximate cloning, quantum purity amplification, and density matrix exponentiation. In each case, coherent processing avoids the measurement-mediated bottleneck that limits incoherent protocols, leading to sample-complexity separations that scale with the local system size.}
  \label{fig:ci_examples}
\end{figure}

We construct an \(n\to m\) random-purification protocol by first applying the random-purification channel of Ref.~\cite{TWZ25}, and then applying the optimal \(n\to m\) cloning channel for pure-state inputs~\cite{W98,KW99}; see~\cref{sec:rp_cloning_supp}.
We show in~\cref{prop:random_purification_cost_functions} that the maximum fidelities for the full output state and for a one-site marginal are lower bounded, respectively, by
\begin{subequations}
\begin{align}
\mathcal F_\mathrm{all}^\ast
&\ge
\frac{\binom{n+dr-1}{dr-1}}{\binom{m+dr-1}{dr-1}}, \\
\mathcal F_\mathrm{one}^\ast
&\ge
\frac{n(m+dr)+m-n}{m(n+dr)}.
\end{align}
\end{subequations}
To facilitate comparison, we set \(m=n+\ell\) with fixed \(\ell\).
We show that one-site worst-case infidelity at most \(\varepsilon\) can be achieved coherently using only \(\bigO{\sqrt{\ell dr/\varepsilon}}\) input copies $n$.
By contrast, for $\varepsilon\leq1/4$, any incoherent protocol requires \(\bigOm[][\varepsilon]{d/\varepsilon}\) copies, even in the pure-state case, since the input state must first be estimated classically.

Tracing out the purifying registers gives mixed-state clones. This produces a coherent protocol for approximate cloning of states of rank at most $r$ with the same $\sqrt{dr}$ versus $d$ separation.

\emph{Example 2: Quantum purity amplification}.\,---
At fixed output number and spectral gap, QPA gives a stronger separation: coherent sample complexity is dimension independent, while the incoherent benchmark scales linearly with $d$.
QPA is the task of coherently processing \(n\) copies of a state \(\rho\), whose \(k\)-th eigenvalue is non-degenerate, to produce \(m\) high-fidelity copies of the corresponding eigenstate, i.e., \(\rho^{\otimes n}\leadsto\psi_k^{\otimes m}\), with $\psi_k$ being the projector onto the $k$-th eigenstate of $\rho$, as illustrated in~\cref{fig:ci_examples}. We write the sorted spectrum of \(\rho\) as \(\boldsymbol p\dot{=}(p_1,\ldots,p_d)\), the spectral gaps as \(D_{k,i}\coloneqq p_k-p_i\), and we set \(D_{k,\mathrm{min}}\coloneqq \min_{i\neq k}|D_{k,i}|\).

In the companion paper~\cite{QPA26}, we develop the overhang-removal protocol for QPA and establish its optimality. For $m=1$, it achieves $n=\bigO{1/(\varepsilon D_{k,\mathrm{min}}^2)}$ for worst-case infidelity $\varepsilon$. Its structure is dictated by symmetry and can be expressed in the following three steps: (1) Schur sampling, which resolves the input into the symmetry sectors labeled by the Young diagram (YD) $\yd{\varsigma}$; (2) a covariant channel on the measured symmetry sector with output in the maximally symmetric irrep $\mathbb{W}^{\yd{m}}$, specified by the overhang-removal rule (see Ref.~\cite[Definition~2]{QPA26}); and (3) tracing out the excess degrees of freedom.
This generalizes the \(m=1\) and \(k=1\) protocol of Ref.~\cite{LFIC24}.

For $m=1$, overhang-removal is uniquely optimal on every symmetry sector $\yd{\varsigma}$ whose row gaps $\Delta_{i,j}\coloneqq \varsigma_i-\varsigma_j$ satisfy $\Delta_{k,\mathrm{min}}\coloneqq \min_{i\neq k}|\Delta_{k,i}|>2/D_{k,\mathrm{min}}$; analogous general-$m$ thresholds are derived in Ref.~\cite[Corollary~S6.4]{QPA26}.
Compared with previous results~\cite{DGHM+25,LFIC24}, our work gives the first optimality characterization with tight constants.
The optimality is robust and holds for various loss functions, including infidelity, trace distance, and cross entropy, owing to the symmetry analyzed in Ref.~\cite[Section~II]{QPA26}.
Compared to Ref.~\cite{DGHM+25} with sample complexity $\bigO{(\varepsilon^{-1}+p_1^{-1})(1-p_1)/D_{1,2}^2}$, our protocol applies to the more general setting of multiple outputs and arbitrary target eigenstates, with improved constant factors and without the additional \(p_1^{-1}\) overhead. Their bound additionally vanishes as $p_1 \to 1$ through a $(1-p_1)$ factor. This corresponds to the trivial pure-state limit, which our sector-wise bound also captures; see Ref.~\cite[Section~II C]{QPA26}.
Circuit-level implementations are given in Ref.~\cite[Section~II D]{QPA26}.

In the EB limit ($m\to\infty$), the one-site marginal of the asymptotically optimal $n\to m$ QPA protocol gives the optimal incoherent protocol on each symmetry sector, specified by the covariant positive operator-valued measure (POVM)
\begin{equation}
\{M_U=d^{\yd{\varsigma}}U^{\yd{\varsigma}}\ketbra{\wt{\mathrm{lw}}}{\wt{\mathrm{lw}}}U^{\yd{\varsigma}\dagger}\mid U\in U(d)\},
\end{equation}
followed by re-preparation of $U\ketbra{k}{k}U^\dagger$, where $\ket{\wt{\mathrm{lw}}}$ denotes the lowest-weight vector in the sector $\mathbb{W}^{\yd{\varsigma}}$. The protocol is asymptotically optimal, achieving sample complexity
\begin{equation}
n
=
\frac{1}{\varepsilon}
\left(\sum\limits_{i\neq k}\frac{p_i}{D_{k,i}^2}
+\sum\limits_{i=k+1}^d\frac{1}{D_{k,i}}
\right)
+\lilo[d,\boldsymbol{p}]{\varepsilon^{-1}}.
\end{equation}
The optimality is exact for $k=1$ and depolarized inputs; see Ref.~\cite[Section~II C]{QPA26}.
By the equivalence in~\cref{subsec:eigenstate_tomography_supp}, this also yields the asymptotically optimal measurement for $k$-th eigenstate tomography.

The above constructions lead to a sample-complexity separation between coherent and incoherent QPA protocols. For fixed \(m\) and \(D_{k,\mathrm{min}}>0\), there exists \(\varepsilon_0=\varepsilon_0(m,D_{k,\mathrm{min}})>0\) such that, to achieve infidelity \(0<\varepsilon\leq\varepsilon_0\), the minimal sample complexity \(n^\ast\) satisfies, by Theorem~\ref{cor:explicit_sample_complexity_inversion} and Corollary~\ref{cor:clean_dmin_sample_complexity},
\begin{equation}
n^\ast
\le
\min\mleft\{
\frac{m}{\varepsilon D_{k,\mathrm{min}}^2}
+
R,
\frac{136m}{\varepsilon D_{k,\mathrm{min}}^2}
\mright\}.
\end{equation}
The remainder satisfies $R<CS\ln(CS)$, where $S=\frac{\sqrt m}{\sqrt\varepsilon D_{k,\mathrm{min}}^2}$ and $C$ is a constant. The threshold $\varepsilon_0$ depends only on $m$ and $D_{k,\mathrm{min}}$; in particular, both the bound and its domain are independent of $d$. In contrast, for $m=1$ and fixed $D_{k,\mathrm{min}}$, the corresponding incoherent protocol requires at least $\frac{d-k}{2\varepsilon}\mleft(1-\lilo[d,k][\varepsilon]{1}\mright)$ samples as $\varepsilon\to0$, scaling linearly in $d-k$ and hence exponentially in the local system size; see~\cref{lem:eb_sample_complexity_lower} for the lower bound, establishing the desired separation. The $d$-dependence indicated in the tail term is an artifact of the inversion; the underlying infidelity bound is uniform in $d$; see~\cref{thm:one_site_risk_lower_bound}.

\emph{Example 3: Density matrix exponentiation}.\,---
The separation persists even when part of the input remains coherent throughout the protocol.
In DME, \(n\) copies of an unknown mixed program state \(\rho\), as illustrated in~\cref{fig:ci_examples}, are used to implement a unitary evolution on a separate quantum data input:
\begin{equation}
{^T}\!\!\mathcal A_\rho(\cdot)
=
e^{-iT\rho}\cdot e^{iT\rho}.
\end{equation}
Here \(T\) is the simulation time. Equivalently, DME can be formulated at the state level by treating both the program state $\rho^{\otimes n}$ and the data state $\sigma_{\mathrm{data}}$ as inputs, as in~\cref{eq:dme_state_level_target}.

The original DME protocol~\cite{LMR14} achieves diamond-norm error \(\varepsilon\) using $n=\bigO{T^2/\varepsilon}$ copies, independent of the local dimension $d$, with nonasymptotic bounds derived in Ref.~\cite{GKPP+25}.

In contrast, a partially coherent protocol in our setting is a classically programmed channel: it first measures $\rho^{\otimes n}$ and then uses the classical outcome to control an arbitrary channel on $\sigma_{\mathrm{data}}$. Previous comparisons considered tomography followed by exponentiation of the estimated state, which does not capture the most general partially coherent strategy~\cite{kimmel2017hamiltonian,ML25}.

To obtain a lower bound for the most general strategy, we restrict the program states to a local $(d-1)$-parameter family of pure states and evaluate the output channel on a fixed data state. We then convert any partially coherent implementation into an estimator for the program state. The resulting estimation lower bound implies that, for fixed $T\not\equiv0\pmod{2\pi}$ and sufficiently small $\varepsilon$, with a threshold depending only on $T$, any partially coherent DME protocol requires $n=\bigOm[][\varepsilon]{\sin^2\left(T/2\right)d/\varepsilon}$ input copies. We measure the error by the worst-case trace distance over all inputs. The lower bound already follows by evaluating a fixed data state, and therefore also holds under the stronger diamond norm error. The full derivation is given in~\cref{lem:embedding_output,thm:dme_partially_coherent_sample_lower}.

\begin{figure}[ht]
\centering
\includegraphics[width=\columnwidth]{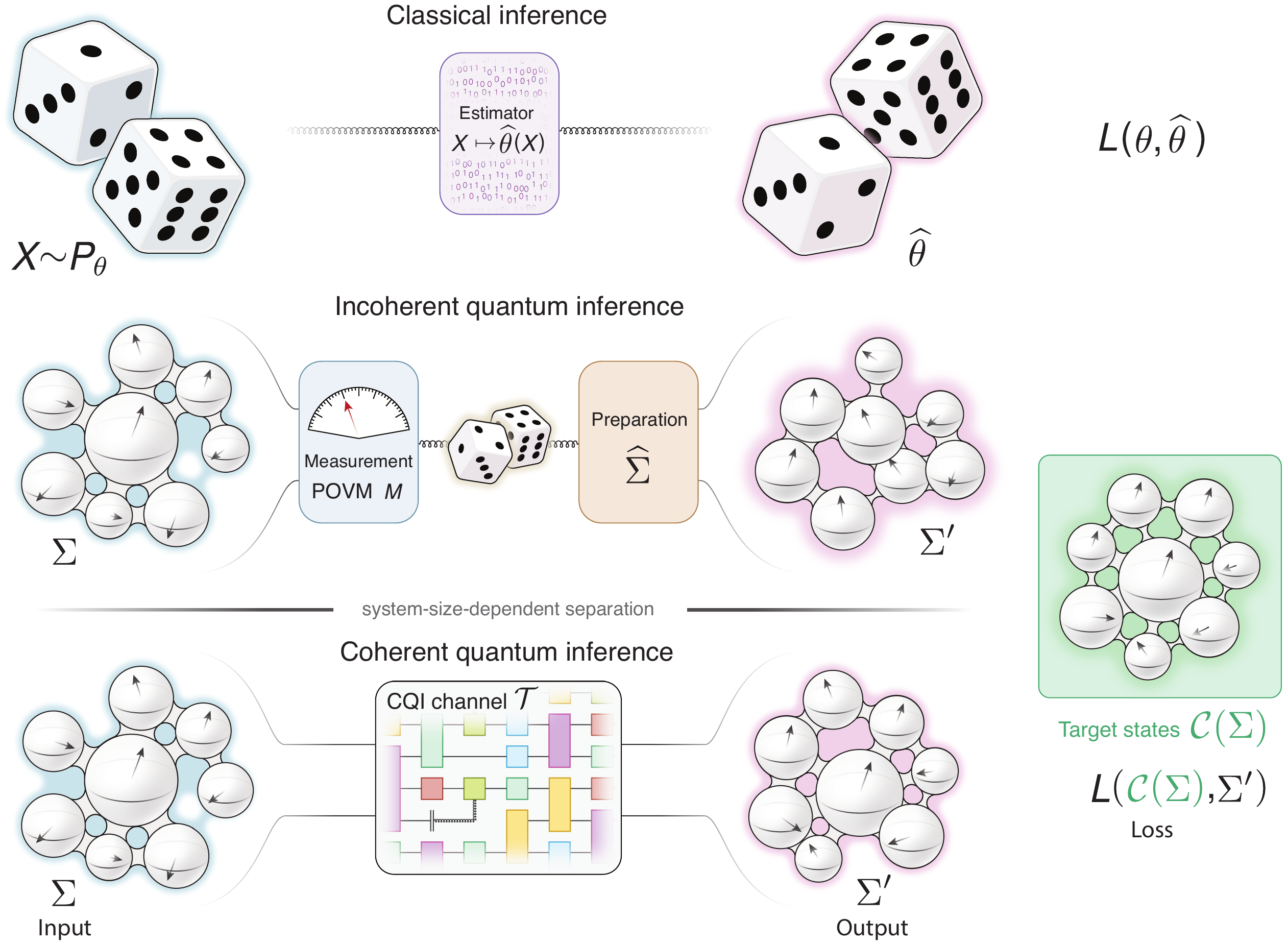}
  \caption{Illustration of classical inference, incoherent quantum inference, and coherent quantum inference. In the classical setting, the random variable \(X\sim P_\theta\) leads to an estimate \(\widehat{\theta}(X)\), with performance judged by a loss \(L\mleft(\theta,\widehat\theta\mright)\). The incoherent quantum analogue measures the input system \(\Sigma\) with a positive operator-valued measure (POVM) \(\{M_i\}\) and then prepares an output state \(\Sigma_i^\prime\), forming a measurement-mediated protocol. In coherent inference, by contrast, the input is processed directly by a quantum channel \(\mathcal{T}\). In both quantum settings, the target set \(\mathcal C(\Sigma)\) depends on the input, and performance is judged by a loss \(L\mleft(\mathcal C(\Sigma),\Sigma^\prime\mright)\).}
  \label{fig:ci_illustration}
\end{figure}

\emph{A framework for coherent inference}\,---
We now bring the three examples of RP, QPA, and DME together and extract their common structure, formalized in~\cref{def:cqi_protocol}, generalizing classical statistical inference theory~\cite{Ber85,LC98} to the setting where the target is a quantum state and the decision rule is a physical channel.
\begin{definition}[Coherent Quantum Inference Protocol]
\label{def:cqi_protocol}
Let $\mathcal{X} \subseteq \mathcal{B}(\mathbb{C}^{d_\mathrm{in}})$ be a set of density matrices on the $d_\mathrm{in}$-dimensional Hilbert space, representing the possible input states.
For each $\Sigma \in \mathcal{X}$, let $\mathcal{C}(\Sigma)\subseteq \mathcal{B}(\mathbb{C}^{d_\mathrm{out}})$ denote the set of target states into which one aims to convert $\Sigma$.
A CQI protocol is a quantum channel
\begin{equation}
\mathcal{T}: \mathcal{B}(\mathbb{C}^{d_\mathrm{in}}) \to \mathcal{B}(\mathbb{C}^{d_\mathrm{out}})
\end{equation}
that, when acting on $\Sigma$, produces an output state $\Sigma^\prime=\mathcal{T}(\Sigma)$.
Performance is quantified by a \emph{loss function} $L\mleft(\mathcal{C}(\Sigma), \Sigma^\prime\mright)$, which measures how well the output captures the desired structure.
The \emph{worst-case risk} of the protocol is defined as
\begin{equation}
\label{eq:cqi_worst_case_risk}
\mathcal{L}_{\mathrm{w}}\mleft(\mathcal{T}\mright)
=
\sup_{\Sigma \in \mathcal{X}}
L\mleft(\mathcal{C}(\Sigma), \mathcal{T}(\Sigma)\mright),
\end{equation}
and given a prior distribution $\pi$ on $\mathcal{X}$, the \emph{average risk} is
\begin{equation}
\label{eq:cqi_average_risk}
\mathcal{L}_{\mathrm{a}}\mleft(\mathcal{T}\mright)
=
\int_{\mathcal{X}}
L\mleft(\mathcal{C}(\Sigma), \mathcal{T}(\Sigma)\mright) \, \mathrm{d}\pi(\Sigma).
\end{equation}
\end{definition}

As illustrated in~\cref{fig:ci_illustration}, a CQI protocol generalizes classical statistical decision theory, with quantum data $\Sigma$, channel output $\mathcal T(\Sigma)$, and target $\mathcal{C}(\Sigma)$ playing the roles of samples $X_1,\ldots,X_n$, estimator $\widehat{\theta}$, and parameter $\theta$. A coherent protocol ($Q\to Q$) processes the input through $\mathcal T$ directly.
Applying standard quantum inference ($Q\to C$) followed by state preparation ($C\to Q$) instead yields an incoherent protocol ($Q\to C\to Q$).
The three examples above---RP, QPA, and DME---are all instances of CQI, with their input family, target-set map, and loss function specified in the SM.

Two complementary directions connect coherent and incoherent inference: in one, coherent protocols reduce to incoherent ones via the EB limit; in the other, classical-output tasks lift to coherent ones by promoting the target into a quantum object.

The EB limit is based on a de Finetti-type argument, with related ideas appearing in pure-state cloning~\cite{BA06} and inductive quantum learning~\cite{MSW17}. Here we extend this principle to general CQI protocols.
Whenever the one-site marginal of a symmetric $1\to m$ CQI protocol converges in diamond norm as $m\to\infty$, the limit is EB and hence corresponds to an incoherent protocol; see~\cref{thm:limit_of_optimal}.
For any well-behaved loss function, if the convergent sequence consists of optimal symmetric protocols, the limiting EB channel attains the limiting optimal one-site risk, which equals the optimal incoherent risk; see~\cref{thm:EB_risk,cor:limit_of_optimal_symmetric}.

The resulting incoherent inference protocols can be further related to standard, classical-output inference protocols, because any EB protocol admits a measurement-mediated representation~\cite{HSR03}, namely,
\begin{equation}
\widetilde{\mathcal T}(\Sigma)
=\langle\widehat{\Sigma}\rangle_M=
\int \Tr\left(M(\mathrm{d}\widehat{\Sigma})\Sigma\right)
\widehat{\Sigma},
\end{equation}
for some POVM $M$ and estimator $\widehat{\Sigma}$. Hence, whenever the loss function is convex, we have
\begin{equation}
L\mleft(\mathcal{C}(\Sigma),\langle\widehat{\Sigma}\rangle_M\mright)
\leq
\left\langle\!
L\mleft(\mathcal{C}(\Sigma),\widehat{\Sigma}\mright)\right\rangle_M.
\end{equation}
Therefore, the optimal EB risk establishes a lower bound on the risk of any POVM for the corresponding standard inference task. In more structured settings, such as QPA, the relevant figure of merit is linear in the output state, so the optimization over incoherent protocols reduces directly to an optimization over the underlying POVM.

In the reverse direction, an incoherent task can be promoted to a coherent task by encoding the estimated property into a quantum object. Suppose the incoherent task estimates a property $f(\Sigma)$ of $\Sigma$. One may define target sets $\mathcal{C}(\Sigma)$ of quantum states that contain all the information of $f(\Sigma)$. Our three examples instantiate this prescription: setting $\mathcal{C}(\rho)=\{\rho^{\otimes m}\}$ lifts state tomography to mixed-state cloning, and QPA and DME similarly arise as coherentizations of eigenstate tomography and full state tomography modulo the simulation period $2\pi/T$.
Whether this prescription leads to a systematic way of quantifying the performance gain from coherence remains an open question; in lieu of a general answer, we illustrate the prescription with several potential use cases:

For spectrum testing, the direct lift is the \emph{coherent diagonalization} task: $f(\rho)=\boldsymbol{p}$ is the spectrum of the state, and we may take $\mathcal{C}(\rho)=\left\{\sum_i p_i\ketbra{i}{i}\right\}$ to be the singleton containing the diagonal state obtained by placing the spectrum in a fixed computational basis.
In this case, the resulting resource state serves as a proxy for unitary-invariant properties such as entropy and purity.
More generally, this framework turns property prediction into coherent tasks indexed by an observable class $\mathcal O$: given a list of target properties $f(\Sigma)$, one may choose $\mathcal{C}(\Sigma)$ to be the singleton containing a state whose expectation values agree with those properties on observables $O\in\mathcal O$.

\emph{Outlook}.\,---
Four directions appear especially important. First, the present examples leave open the problem of sharp rates: one would like precise coherent and incoherent bounds for cloning, RP, and DME, together with improved nonasymptotic and multi-output bounds for QPA. Second, it remains to understand whether coherent sample-complexity advantages generally translate into implementation advantages. While efficient coherent implementations are known for QPA and DME, the corresponding question remains open for RP and cloning. Third, the reverse direction should be developed systematically, clarifying when coherentizing a classical inference task reduces sample complexity and how large this reduction can be in more general settings. Fourth, CQI should be extended from state-valued outputs to channel-valued outputs~\cite{MSW17}, or even to higher-order quantum operations, connecting the framework to tasks such as quantum signal processing, Hamiltonian simulation, quantum learning of transformations~\cite{bisio2010optimal}, and fault-tolerant resource conversion.

\section{Acknowledgments}
We thank A. Gilyén for helpful discussions.
Z.L., A.W.H. and I.L.C. acknowledge support from the U.S. Department of Energy, Office of Science, National Quantum Information Science Research Centers, Co-design Center for Quantum Advantage (C\(^2\)QA), under contract No. DE-SC0012704.
This work is supported by the National Science Foundation under Cooperative Agreement PHY-2019786 (The NSF AI Institute for Artificial Intelligence and Fundamental Interactions).
ET is supported by the European Union via an ERC grant (QInteract, Grant No 101078107) and VILLUM FONDEN (Grant No 10059 and 37532).
The authors further acknowledge that they used generative AI tools for writing and editing, literature review, and assistance with proofs of some lemmas and theorems. The main conceptual and technical contributions originated with the authors. In particular, the coherent quantum inference framework and the main proof establishing the coherent--incoherent sample-complexity separation were developed without AI assistance.
\putbib[ci_main.bib]
\end{bibunit}

\beginsupplement
\begin{bibunit}[apsrev4-2]
\makeatletter
\patchcmd{\NAT@bibsetnum}{LastBibItem}{LastBibItemSupplement}{}{}
\patchcmd{\endthebibliography}{LastBibItem}{LastBibItemSupplement}{}{}
\makeatother
\begin{center}
  {\Large \textbf{Supplemental Material}}\\[0.25cm]
  {\normalsize for}\\[0.2cm]
  {\large \textbf{An Exponential Sample-Complexity Advantage}}\\[0.08cm]
  {\large \textbf{for Coherent Quantum Inference}}\\[0.35cm]
  {\normalsize Zhaoyi Li,\textsuperscript{1,2,*} Elias Theil,\textsuperscript{3,*} Aram W. Harrow,\textsuperscript{1} and Isaac Chuang\textsuperscript{1}}\\[0.25cm]

  {\footnotesize\itshape
  \textsuperscript{1}Department of Physics, Massachusetts Institute of Technology, Cambridge MA 02139, USA \\
  \textsuperscript{2}The NSF AI Institute for Artificial Intelligence and Fundamental Interactions \\
  \textsuperscript{3}Department of Mathematical Sciences, University of Copenhagen, 2100 Copenhagen, Denmark \\
  \textsuperscript{*}These authors contributed equally. Emails: \texttt{ladmon@mit.edu}, \texttt{edmt@math.ku.dk}.}
\end{center}

\vspace{0.5cm}
\noindent In this Supplemental Material (SM), we provide technical details and proofs supporting the main text.
We first establish general properties of coherent quantum inference (CQI) risk and symmetry reduction in \cref{sec:risk_symmetry_supp}: convexity and continuity of the risk functional are proved in \cref{subsec:convexity_supp}, the one-site and all-site risks are defined in \cref{subsec:one_all_site_risks}, and the relevant twirling reductions are given in \cref{subsec:symmetries_supp}.
We then relate coherent protocols to their entanglement-breaking (EB) limits in \cref{sec:coh_incoh_supp}, by proving approximate separability for symmetric extendible Choi matrices (\cref{subsec:approx_sep_supp}) and deriving the corresponding EB limit of optimal symmetric channels (\cref{subsec:eb_limit_supp}).
The remaining sections apply these tools to the three examples discussed in the main text: random purification (RP) and approximate cloning are treated in \cref{sec:rp_cloning_supp}, quantum purity amplification (QPA) and eigenstate tomography in \cref{sec:qpa_supp}, and density matrix exponentiation (DME) in \cref{sec:dme_supp}.

\section*{Contents}
\vspace{0.3cm}
\begin{enumerate}[label=S\arabic*.]
    \item \textit{General framework for coherent inference}\dotfill \pageref{sec:risk_symmetry_supp}
        \begin{enumerate}[label=\Alph*.]
            \item \textit{Convexity and continuity of CQI risk}\dotfill \pageref{subsec:convexity_supp}
            \item \textit{One-site and all-site risks}\dotfill \pageref{subsec:one_all_site_risks}
            \item \textit{Symmetries and twirling}\dotfill \pageref{subsec:symmetries_supp}
        \end{enumerate}
    \item \textit{Connecting coherent and incoherent protocols}\dotfill \pageref{sec:coh_incoh_supp}
        \begin{enumerate}[label=\Alph*.]
            \item \textit{Approximate separability of symmetric extendible Choi matrices}\dotfill \pageref{subsec:approx_sep_supp}
            \item \textit{EB limit}\dotfill \pageref{subsec:eb_limit_supp}
        \end{enumerate}
    \item \textit{RP and approximate cloning}\dotfill \pageref{sec:rp_cloning_supp}
        \begin{enumerate}[label=\Alph*.]
            \item \textit{Coherent RP protocol}\dotfill \pageref{subsec:coherent_rp_protocol}
            \item \textit{Incoherent benchmark}\dotfill \pageref{subsec:rp_eb_benchmark}
        \end{enumerate}
    \item \textit{QPA and eigenstate tomography}\dotfill \pageref{sec:qpa_supp}
        \begin{enumerate}[label=\Alph*.]
            \item \textit{Nonasymptotic QPA sample-complexity bounds}\dotfill \pageref{subsec:qpa_sample_complexity_supp}
            \item \textit{Eigenstate tomography}\dotfill \pageref{subsec:eigenstate_tomography_supp}
            \item \textit{POVM of the EB channel}\dotfill \pageref{subsec:eb_povm_supp}
            \item \textit{Optimal EB risk}\dotfill \pageref{subsec:optimal_eb_risk_supp}
        \end{enumerate}
    \item \textit{DME}\dotfill \pageref{sec:dme_supp}
\end{enumerate}
\clearpage

\section{General Framework for Coherent Inference}
\label{sec:risk_symmetry_supp}

\subsection{\texorpdfstring{Convexity and continuity of CQI risk}{Convexity and continuity of CQI risk}}
\label{subsec:convexity_supp}

We call a loss $\ell\mleft(\Gamma,\Xi\mright)$ that directly compares two states an \emph{elementary loss}; standard examples are trace distance $\frac{1}{2}\|\Gamma-\Xi\|_1$, squared Bures distance $2(1-\sqrt{F(\Gamma,\Xi)})$, and infidelity $1-F(\Gamma,\Xi)$.

We begin by establishing the convexity of the risk functional with respect to the quantum channel. The following lemma shows that the set-level loss $L$ is convex when the elementary loss $\ell$ is jointly convex and the target set is convex.

\begin{lemma}[Convexity of the set-level loss]
\label{lem:induced_loss_convexity}
Let $\mathcal C$ be a nonempty convex set of target states, and let the elementary loss $\ell$ be jointly convex. Then the induced loss
\begin{equation}
\label{eq:induced_set_loss}
L\mleft(\mathcal C,\Xi\mright)
\coloneqq
\inf_{\Gamma\in\mathcal C}\ell\mleft(\Gamma,\Xi\mright)
\end{equation}
is convex in $\Xi$.
\end{lemma}

\begin{proof}
For any states $\Xi_1,\Xi_2$, any $\lambda\in[0,1]$, and any $\delta>0$, choose $\Gamma_1,\Gamma_2\in\mathcal C$ such that $\ell\mleft(\Gamma_i,\Xi_i\mright)\le L\mleft(\mathcal C,\Xi_i\mright)+\delta$ for $i=1,2$. Since $\mathcal C$ is convex, $\lambda\Gamma_1+(1-\lambda)\Gamma_2\in\mathcal C$. Joint convexity of $\ell$ therefore gives
\begin{equation}
\label{eq:induced_set_loss_convexity}
\begin{aligned}
L\mleft(\mathcal C,\lambda\Xi_1+(1-\lambda)\Xi_2\mright)
&\le
\ell\mleft(\lambda\Gamma_1+(1-\lambda)\Gamma_2,
\lambda\Xi_1+(1-\lambda)\Xi_2\mright)\\
&\le
\lambda L\mleft(\mathcal C,\Xi_1\mright)
+(1-\lambda)L\mleft(\mathcal C,\Xi_2\mright)+\delta.
\end{aligned}
\end{equation}
Taking $\delta\to0$ proves the claim.
\end{proof}

\begin{lemma}[Convexity of risk]
\label{lem:risk_convexity}
Let $\mathbb{H}_\mathrm{in}$ and $\mathbb{H}_\mathrm{out}$ be the input and output Hilbert spaces of the coherent inference protocol.
Let $\mathcal{X}\subseteq\mathcal B(\mathbb{H}_\mathrm{in})$ be a compact family of input states.
Let $\pi$ be a probability measure on $\mathcal X$.
Let $\mathcal C(\Sigma)$, $\ell$, and $L$ satisfy the assumptions of~\cref{lem:induced_loss_convexity} for all $\Sigma\in\mathcal X$.
Define the worst-case and average risk functionals, respectively, by
\begin{equation}
\mathcal L_{\mathrm{w}}\mleft(\mathcal T\mright)
=
\sup_{\Sigma\in\mathcal X}
L\mleft(\mathcal C(\Sigma),\,\mathcal T(\Sigma)\mright),
\end{equation}
and
\begin{equation}
\mathcal L_{\mathrm{a}}\mleft(\mathcal T\mright)
=
\int_{\mathcal X} L\mleft(\mathcal C(\Sigma),\,\mathcal T(\Sigma)\mright)\,\mathrm{d}\pi(\Sigma).
\end{equation}
Both are convex in the channel $\mathcal T$.
That is, for any two channels $\mathcal T_1,\mathcal T_2$ and any $\lambda\in[0,1]$,
\begin{equation}
\mathcal L_{(\cdot)}\mleft(\lambda\mathcal T_1+(1-\lambda)\mathcal T_2\mright)
\le
\lambda\,\mathcal L_{(\cdot)}\mleft(\mathcal T_1\mright)
+(1-\lambda)\,\mathcal L_{(\cdot)}\mleft(\mathcal T_2\mright),
\end{equation}
where $(\cdot)$ is $\mathrm w$ or $\mathrm a$.
\end{lemma}

\begin{proof}
For any $\Sigma\in\mathcal X$,
\begin{equation}
(\lambda\mathcal T_1+(1-\lambda)\mathcal T_2)\left(\Sigma\right)
=
\lambda\mathcal T_1\left(\Sigma\right)
+(1-\lambda)\mathcal T_2\left(\Sigma\right).
\end{equation}
By convexity of the loss in its second argument,
\begin{equation}
L\mleft(\mathcal C(\Sigma),(\lambda\mathcal T_1+(1-\lambda)\mathcal T_2)\mleft(\Sigma\mright)\mright)
\le
\lambda L\mleft(\mathcal C(\Sigma),\mathcal T_1\mleft(\Sigma\mright)\mright)
+(1-\lambda)L\mleft(\mathcal C(\Sigma),\mathcal T_2\mleft(\Sigma\mright)\mright).
\end{equation}
Taking the supremum over $\Sigma\in\mathcal X$ and using the subadditivity of the supremum, we obtain the stated inequality. The same pointwise inequality, followed by integration, gives the result for the average risk $\mathcal L_{\mathrm{a}}\mleft(\mathcal T\mright)$.
\end{proof}

\begin{lemma}[Risk continuity]
\label{lem:risk_continuity}
Assume that the loss function $L\mleft(\mathcal C(\Sigma),\Xi\mright)$ is Lipschitz in its second argument $\Xi$ with constant $L_0$, uniformly over all $\Sigma$. We call such losses well behaved. Then the risk functional $\mathcal L_{\mleft(\cdot\mright)}\mleft(\cdot\mright)$ is Lipschitz in the channel with constant $L_0$ with respect to the trace norm on Choi operators, for both the worst-case and average risks. More precisely, if $S$ and $T$ are the Choi operators of channels $\mathcal S$ and $\mathcal T$, then
\begin{equation}
\left|\mathcal L_{\mleft(\cdot\mright)}\mleft(\mathcal S\mright)-\mathcal L_{\mleft(\cdot\mright)}\mleft(\mathcal T\mright)\right|
\le
L_0 \|S-T\|_1.
\end{equation}
\end{lemma}

\begin{proof}
Suppose that
$\|S-T\|_1\le \varepsilon.$
Consider the corresponding channels $\mathcal S$ and $\mathcal T$. For every $\Sigma$, we have
\begin{equation}
\mathcal S(\Sigma)-\mathcal T(\Sigma)
=\Tr_A\!\Big(\big(\Sigma_A^\top \otimes I_B\big)(S-T)\Big).
\end{equation}
Since $\mleft(\mathcal S-\mathcal T\mright)(\Sigma)$ is Hermitian,
\begin{equation}
\begin{aligned}
\bigl\|\mleft(\mathcal S-\mathcal T\mright)(\Sigma)\bigr\|_1
&=
\sup_{\|H\|_\infty\le 1}
\left|
\Tr\mleft[H\mleft(\mathcal S-\mathcal T\mright)(\Sigma)\mright]
\right|\\
&=
\sup_{\|H\|_\infty\le 1}
\left|
\Tr\mleft[\mleft(\Sigma^{\mathsf T}\otimes H\mright)(S-T)\mright]
\right|\\
&\le
\sup_{\|H\|_\infty\le 1}
\bigl\|\Sigma^{\mathsf T}\otimes H\bigr\|_\infty
\|S-T\|_1\\
&\le
\|S-T\|_1.
\end{aligned}
\end{equation}
Here the last inequality follows from $\|\Sigma^{\mathsf T}\|_\infty\le 1$. Therefore, by the assumed Lipschitz property of $L$, for all $\Sigma$,
\begin{equation}
\bigl|
L\mleft(\mathord{\cdot},\mathcal S(\Sigma)\mright)
-
L\mleft(\mathord{\cdot},\mathcal T(\Sigma)\mright)
\bigr|
\le
L_0\,\varepsilon.
\end{equation}
For the average risk \(\mathcal L_{\mathrm{a}}\), integrating over the probability measure over the inputs \(\Sigma\) gives the desired inequality. For the worst-case risk \(\mathcal L_{\mathrm{w}}\), we take the supremum, which yields a similar bound.
\end{proof}

\subsection{One-site and all-site risks}
\label{subsec:one_all_site_risks}

When the target set is generated by tensor-power states, i.e., there exists a set $c(\Sigma)\subseteq\mathcal{B}(\mathbb{H}_{\mathrm{out}_1})$ such that $\mathcal{C}(\Sigma)=\mathrm{conv}\left\{\gamma^{\otimes m}\mid\gamma\in c(\Sigma)\right\}$, and the output space correspondingly factorizes as $\mathbb{H}_{\mathrm{out}}=\mathbb{H}_{\mathrm{out}_1}\otimes\cdots\otimes\mathbb{H}_{\mathrm{out}_m}$,
there are two alternative ways to define the induced loss from an elementary loss $\ell$. Applying $\ell$ directly to the joint target and output gives the all-site loss:
\begin{equation}
\begin{aligned}
L_{\mathrm{all}}\mleft(\mathcal C(\Sigma),\Xi\mright)
&\coloneqq
\inf_{\Gamma\in\mathcal C(\Sigma)}
\ell\mleft(\Gamma,\Xi\mright).
\end{aligned}
\end{equation}
Since $\Tr_{i^c}\mathcal C(\Sigma)=\mathrm{conv}\mleft(c(\Sigma)\mright)$ for every $i$, restricting the comparison to one-site marginals gives the one-site loss:
\begin{equation}
\begin{aligned}
L_{\mathrm{one}}\mleft(\mathcal C(\Sigma),\Xi\mright)
&\coloneqq
\inf_{\gamma\in \mathrm{conv}\mleft(c(\Sigma)\mright)}
\max_i
\ell\mleft(
\gamma,
\Tr_{\mathrm{out}_{i^\complement}}
\mleft(\Xi\mright)
\mright).
\end{aligned}
\end{equation}

{
\begin{lemma}[Convexity of the one-site loss]
\label{lem:one_site_loss_convexity}
For every input $\Sigma$, if the elementary loss $\ell$ is jointly convex, then $L_{\mathrm{one}}\mleft(\mathcal C(\Sigma),\Xi\mright)$ is convex in the output state $\Xi$. Consequently, the corresponding one-site worst-case and average risk functionals are convex in the channel $\mathcal T$.
\end{lemma}

\begin{proof}
For each $i$, define
\begin{equation*}
f_i(\gamma,\Xi)
\coloneqq
\ell\mleft(\gamma,\Tr_{\mathrm{out}_{i^\complement}}\mleft(\Xi\mright)\mright).
\end{equation*}
The partial trace is linear, so joint convexity of $\ell$ implies that each $f_i$ is jointly convex in $\gamma$ and $\Xi$. Therefore, for $\gamma=\lambda\gamma_1+(1-\lambda)\gamma_2$ and $\Xi=\lambda\Xi_1+(1-\lambda)\Xi_2$,
\begin{equation*}
\max_i f_i(\gamma,\Xi)
\le
\max_i\mleft(\lambda f_i(\gamma_1,\Xi_1)+(1-\lambda)f_i(\gamma_2,\Xi_2)\mright)
\le
\lambda\max_i f_i(\gamma_1,\Xi_1)
+(1-\lambda)\max_i f_i(\gamma_2,\Xi_2).
\end{equation*}
Thus the maximum over $i$ is jointly convex. Minimizing over the convex set $\mathrm{conv}\mleft(c(\Sigma)\mright)$ preserves convexity by the same argument. The supremum and integral defining the worst-case and average risks then preserve convexity, as in~\cref{lem:risk_convexity}.
\end{proof}
}

Substituting these losses into the definitions in~\cref{eq:cqi_worst_case_risk,eq:cqi_average_risk} gives the corresponding all-site and one-site worst-case and average risks.
\subsection{Symmetries and Twirling}
\label{subsec:symmetries_supp}
We exploit the symmetries of the inference problem to simplify the optimization of CQI protocols. To do so, we consider the most general setting with group symmetries. Let $\mathcal T$ be a CQI protocol, and let $G$ be a symmetry group with unitary representations $U_{g,\mathrm{in}}\in \mathcal{B}(\mathbb{H}_{\mathrm{in}})$ and $U_{g,\mathrm{out}}\in \mathcal{B}(\mathbb{H}_{\mathrm{out}})$.
Define the $G$-twirl superchannel by $\boldsymbol{\mathcal T}_G(\mathcal T)\coloneqq \int_G \mathcal U_{g^{-1},\mathrm{out}}\circ \mathcal T \circ \mathcal U_{g,\mathrm{in}}\,\mathrm{d}g$, with $\mathcal U_{g,\mathrm{in}}(\cdot)=U_{g,\mathrm{in}} \cdot U_{g,\mathrm{in}}^\dagger$ and $\mathcal U_{g,\mathrm{out}}(\cdot)=U_{g,\mathrm{out}} \cdot U_{g,\mathrm{out}}^\dagger$.

For a target set $\mathcal C$, define $\mathcal U_{g,\mathrm{out}}(\mathcal C)\coloneqq\{\mathcal U_{g,\mathrm{out}}(\Gamma)\mid\Gamma\in\mathcal C\}$. The loss $L$ is $G$-invariant if
\begin{equation*}
L\mleft(\mathcal U_{g,\mathrm{out}}(\mathcal C),\mathcal U_{g,\mathrm{out}}(\Xi)\mright)=L\mleft(\mathcal C,\Xi\mright),
\end{equation*}
and the target-set map $\mathcal C$ is $G$-covariant if
\begin{equation*}
\mathcal C\mleft(\mathcal U_{g,\mathrm{in}}(\Sigma)\mright)=\mathcal U_{g,\mathrm{out}}\mleft(\mathcal C(\Sigma)\mright).
\end{equation*}

\begin{lemma}[Risk monotonicity under group twirling]
\label{thm:risk_group_twirling}
Assume that the risk functional $\mathcal L_{(\cdot)}$ is convex, that the loss function $L$ is $G$-invariant, and that the target-set map $\mathcal C$ is $G$-covariant. For the worst-case risk, assume that the input family $\mathcal X$ is $G$-invariant; for the average risk, assume instead that the prior $\pi$ is $G$-invariant. Then the corresponding risk is non-increasing under group twirling, namely $\mathcal L_{(\cdot)}\mleft(\boldsymbol{\mathcal T}_G(\mathcal T)\mright)\le \mathcal L_{(\cdot)}\mleft(\mathcal T\mright)$.
\end{lemma}

\begin{proof}
By the assumed convexity of the risk functional, we have \begin{equation}\mathcal L_{(\cdot)}\mleft(\boldsymbol{\mathcal T}_G(\mathcal T)\mright)\le \int_G \mathcal L_{(\cdot)}\mleft(\mathcal U_{g^{-1},\mathrm{out}}\circ \mathcal T \circ \mathcal U_{g,\mathrm{in}}\mright)\,\mathrm{d}g.\end{equation}
For each $g\in G$, the $G$-invariance of $L$, together with the $G$-covariance of $\mathcal C$, implies that evaluating the loss of $\mathcal U_{g^{-1},\mathrm{out}}\circ \mathcal T \circ \mathcal U_{g,\mathrm{in}}$ is equivalent to evaluating $L\mleft(\mathcal C(\mathcal U_{g,\mathrm{in}}(\Sigma)), \mathcal T(\mathcal U_{g,\mathrm{in}}(\Sigma))\mright)$. Thus, in the worst-case setting, the loss becomes $\sup_{\Sigma\in \mathcal X}L\mleft(\mathcal C(\mathcal U_{g,\mathrm{in}}(\Sigma)), \mathcal T(\mathcal U_{g,\mathrm{in}}(\Sigma))\mright)$, while in the average setting it becomes $\langle L\mleft(\mathcal C(\mathcal U_{g,\mathrm{in}}(\Sigma)), \mathcal T(\mathcal U_{g,\mathrm{in}}(\Sigma))\mright)\rangle_\pi$. Since the input family $\mathcal X$ in the worst-case setting, or the prior in the average setting, is assumed to be $G$-invariant, the rotation can be absorbed into the supremum or the average, yielding $\mathcal L_{(\cdot)}\mleft(\mathcal U_{g^{-1},\mathrm{out}}\circ \mathcal T \circ \mathcal U_{g,\mathrm{in}}\mright)=\mathcal L_{(\cdot)}\mleft(\mathcal T\mright)$ for all $g$. Therefore $\mathcal L_{(\cdot)}\mleft(\boldsymbol{\mathcal T}_G(\mathcal T)\mright)\le \mathcal L_{(\cdot)}\mleft(\mathcal T\mright)$.
\end{proof}

The $G$-twirl projects the set of channels onto the subspace of \(G\)-covariant channels, which we refer to generically as symmetric protocols depending on the context. By~\cref{thm:risk_group_twirling}, the minimal risk is always attained within this class.

We next consider two symmetry classes that are especially relevant in practice and arise naturally when the input and output spaces admit tensor-product structure, as in the one-site and all-site setting of~\cref{subsec:one_all_site_risks}.

{
\begin{theorem}[Exchange-invariant and unitary-covariant twirling reductions]
\label{thm:tensor_product_twirling_reductions}
\color{black}\pdfrender{TextRenderingMode=0}
\textit{Exchange invariance.} On the input side, let the permutation group $S_n$ act through permutation unitaries $P_\pi$, $\pi\in S_n$, on the $n$ input registers. On the output side, let $S_m$ act through permutation unitaries $P_\pi$, $\pi\in S_m$, on the $m$ output registers. Define the exchange channel $\mathrm{exc}:\mathcal B(\mathbb{H}^{\otimes n})\to\mathcal B(\mathbb{H}^{\otimes n})$ by
\[
\mathrm{exc}(\cdot)=\frac{1}{n!}\sum_{\pi\in S_n}P_\pi\,\cdot\,P_\pi^\dagger,
\]
and define the input permutation twirling by $\boldsymbol{\mathcal P}_{\mathrm{in}}(\mathcal T)=\mathcal T\circ \mathrm{exc}_{\mathrm{in}}$ and the output permutation twirling by $\boldsymbol{\mathcal P}_{\mathrm{out}}(\mathcal T)=\mathrm{exc}_{\mathrm{out}}\circ \mathcal T$. If $\mathcal L_{(\cdot)}$ is convex and $L$ and $\mathcal C$ are exchange-invariant under $S_n\times S_m$, then $\boldsymbol{\mathcal P}_{\mathrm{in}}$ and $\boldsymbol{\mathcal P}_{\mathrm{out}}$ do not increase the worst-case risk when $\mathcal X$ is exchange-invariant, and do not increase the average risk when $\pi$ is exchange-invariant.

In particular, suppose that every admissible input $\Sigma$ has the form $\rho^{\otimes n}$, the corresponding target set $\mathcal C(\Sigma)$ is the convex hull of states $\gamma^{\otimes m}$, and $L$ is exchange invariant. Then $\boldsymbol{\mathcal P}_{\mathrm{in}}$ does not change either the worst-case or average risk, and neither does any individual output permutation entering $\boldsymbol{\mathcal P}_{\mathrm{out}}$, without any convexity assumption.

\smallskip
\textit{Unitary covariance.} Let $U(d)$ act in the natural tensor-power representations on the input and output spaces, by conjugating the input by $U^{\otimes n}$ and the output by $U^{\otimes m}$. Define the unitary-twirled superchannel $\boldsymbol{\mathcal U}$ by averaging over the Haar measure on $SU(d)$:
\[
\boldsymbol{\mathcal U}(\mathcal T)(\cdot)
=
\int_{SU(d)} (U^\dagger)^{\otimes m}\, \mathcal T\!\left( U^{\otimes n} (\cdot) (U^\dagger)^{\otimes n} \right)\, U^{\otimes m}\,\mathrm{d}U.
\]
If $\mathcal L_{(\cdot)}$ is convex, $L$ is unitarily invariant, and $\mathcal C$ is unitarily covariant, then $\boldsymbol{\mathcal U}$ does not increase the worst-case risk when $\mathcal X$ is invariant under conjugation by $U(d)$, and does not increase the average risk when $\pi$ is invariant under conjugation by $U(d)$.
\end{theorem}
}

\begin{proof}
The two twirling statements follow from~\cref{thm:risk_group_twirling}, with $G=S_n\times S_m$ and $G=U(d)$, respectively. For the tensor-power specialization, $P_\pi\rho^{\otimes n}P_\pi^\dagger=\rho^{\otimes n}$, and every tensor-power target, hence every convex combination of such targets, is invariant under output permutations. Permutation invariance of $L$ therefore gives equality pointwise for every admissible input, and hence for both risks.
\end{proof}

This is the relevant setting for tasks such as QPA and cloning, where the admissible family $\mathcal{X}$ consists of tensor-product inputs of the form $\rho^{\otimes n}$, and the targets are of the form $\gamma(\rho)^{\otimes m}$, with $\gamma(\rho)$ determined by a single input copy $\rho$.

Another fact about unitary covariant channels, which will be useful for characterizing the form of the channel outputs, is that they preserve the eigenbasis:
\begin{lemma}
\label{lem:unitary_covariance_preserves_eigenbasis}
Consider the eigenbasis \(\{|i\rangle\}_{i=1}^d\) of the density matrix $\rho$. Then the output $\Sigma'\coloneqq\mathcal T(\rho^{\otimes n})$ of the unitary-covariant channel admits a block-diagonal form under the unitary transformation corresponding to type decomposition:
\begin{equation}
\Sigma^\prime
\sim
\bigoplus_{\boldsymbol\tau} X_{\boldsymbol\tau},
\quad
X_{\boldsymbol\tau}\in \mathcal B(\mathbb{T}_{\boldsymbol\tau}).
\end{equation}
Here the type of a computational basis tensor is its occupation vector \(\boldsymbol\tau\), where \(\tau_i\) counts the number of entries equal to \(i\). Thus \(\mathbb{T}_{\boldsymbol\tau}\) is the subspace spanned by all computational basis tensors of type \(\boldsymbol\tau\), and the output space decomposes as $\mathbb{H}^{\otimes m}\cong\bigoplus_{\boldsymbol\tau}\mathbb{T}_{\boldsymbol\tau}$.
When restricting to the $\mathbb{W}^{\yd{1,\ldots,m}}$ subspace of the output space, it follows that $\Sigma^\prime$ is diagonal in the Schur basis.
\end{lemma}

\begin{proof}
For every diagonal unitary $D_{\boldsymbol\theta}=\sum_i e^{i\theta_i}|i\rangle\langle i|,$
we have \(D\rho D^\dagger=\rho\) (omitting the $\boldsymbol\theta$ subscript). Hence, by unitary covariance,
\begin{equation}
[D^{\otimes m},\,\mathcal T(\rho^{\otimes n})]=0.
\end{equation}
Now each type space \(\mathbb{T}_{\boldsymbol\tau}\) carries a distinct character $\chi_{\boldsymbol\tau}(\boldsymbol\theta)=e^{i\boldsymbol\tau\cdot\boldsymbol\theta}$ of the diagonal torus \(T\subseteq U(d)\), so by Schur's lemma any operator commuting with all \(D^{\otimes m}\) must preserve each \(\mathbb{T}_{\boldsymbol\tau}\). Therefore \(\mathcal T(\rho^{\otimes n})\) cannot couple different type spaces.

Finally, when restricting to the totally-symmetric subspace, each \(\mathbb{T}_{\boldsymbol\tau}\) is one dimensional, so \(\Sigma^\prime\) is diagonal in the Schur basis.
\end{proof}

For symmetric protocols, the relation between the worst-case and average risks becomes particularly simple when the admissible family is a single unitary orbit.

\begin{lemma}[Worst-case--average equivalence on a single orbit]
\label{lem:worstcase_bayes_single_orbit}
Assume that for some input state $\Sigma$, \(\mathcal X=\{U\Sigma U^\dagger \mid U\in G\}\). If $\mathcal C$ is \(G\)-covariant and \(L\) is \(G\)-invariant, then the loss of the twirled protocol is constant over \(\mathcal{X}\) and we have
\begin{equation}
    \mathcal L_{\mathrm{w}}\mleft(\boldsymbol{\mathcal T}_G(\mathcal T)\mright) = \mathcal L_{\mathrm{a}}\mleft(\boldsymbol{\mathcal T}_G(\mathcal T)\mright) = L\mleft(\mathcal C(\Sigma),\boldsymbol{\mathcal T}_G(\mathcal T)(\Sigma)\mright) \, .
\end{equation}
\end{lemma}

\begin{proof}
    We have
    \begin{equation}
        L\mleft(\mathcal C(\Sigma),\boldsymbol{\mathcal T}_G(\mathcal T)(\Sigma)\mright) = L\mleft(U\mathcal C(\Sigma)U^\dagger,U\boldsymbol{\mathcal T}_G(\mathcal T)(\Sigma)U^\dagger\mright) = L\mleft(\mathcal C(U\Sigma U^\dagger),\boldsymbol{\mathcal T}_G(\mathcal T)(U\Sigma U^\dagger)\mright)
    \end{equation}
    for all \(U\in G\). Since \(\mathcal X=\{U\Sigma U^\dagger \mid U\in G\}\) the statement follows.
\end{proof}

\section{Connecting Coherent and Incoherent Protocols}
\label{sec:coh_incoh_supp}

\subsection{Approximate Separability of Symmetric Extendible Choi Matrices}
\label{subsec:approx_sep_supp}

We recast the comparison between coherent and incoherent protocols in terms of their Choi matrices. This lets us apply finite de Finetti-type reasoning: a Choi matrix with a sufficiently large symmetric extension is close to one that is separable across the input-output cut. Since separable Choi matrices correspond to incoherent protocols, this links symmetric coherent protocols to incoherent ones.

\begin{theorem}
\label{thm:approx_sep_extendible}
Let $\mathbb{H}_A=\mathbb{C}^{d_A}$ and $\mathbb{H}_{B_i}=\mathbb{C}^{d_B}$ for $i\in\{1,\ldots,m\}$, and let $T_{A B_1}$ be a Choi matrix on $\mathbb{H}_A \otimes \mathbb{H}_{B_1}$, i.e., positive semidefinite (PSD) with $T_A = \Tr_{B_1} T_{A B_1} = I_A$.
Assume moreover that $T_{A B_1}$ is $m$-extendible, i.e., there exists a Choi matrix
\begin{equation}
  T_{A B} \in \mathcal B\bigl(\mathbb{H}_A \otimes \mathbb{H}_B^{\otimes m}\bigr),
  \qquad B = B_1 \cdots B_m,
\end{equation}
such that
\begin{enumerate}
  \item $T_{A B_1} = \Tr_{B_2 \cdots B_m} T_{A B}$, and
  \item $T_{A B}$ is exchangeable in $B$, i.e.,
  \begin{equation}
    \mathrm{exc}_B\mleft(T_{A B}\mright)
    = T_{A B}.
  \end{equation}
\end{enumerate}
Then there exists a Choi matrix $\widetilde{T}_{A B_1} \in \mathcal B(\mathbb{H}_A \otimes \mathbb{H}_{B_1})$ which is separable between $A$ and $B_1$, $\tr_{B_1}\widetilde{T}_{AB_1}=I_A$, and
\begin{equation}
  \bigl\| T_{A B_1} - \widetilde{T}_{A B_1} \bigr\|_1 \le  \frac{2 d_A d_B^2}{m}.
\end{equation}
If, in addition, the extension is supported on the totally symmetric subspace \(\mathbb{W}_B^{\yt{1,\ldots,m}}\subseteq\mathbb{H}_B^{\otimes m}\), then the same conclusion holds with the improved bound
\begin{equation}
  \bigl\| T_{A B_1} - \widetilde{T}_{A B_1} \bigr\|_1 \le  \frac{2 d_A d_B}{m}.
\end{equation}
\end{theorem}
\begin{proof}
By assumption, $\frac{T_{A B}}{d_A}$ is exchange symmetric in $B$, so we may apply the de Finetti theorem to its symmetric purification, $\psi_{A A' B B'}$; i.e., applying the projector onto the fully symmetric subspace gives $\Pi^{\yd{m}}_{BB'}\,\psi_{A A' B B'} = \psi_{A A' B B'}$.
By Theorem~II.3 of Ref.~\cite{CKMR07S}, we can construct the following measurement-mediated channel
$\mathrm{MP}_{CC' \to B_1 B_1'}$ such that
\begin{equation}
  \label{eq:CKMR_main}
  \bigl\|
    \psi_{A A^\prime B_1 B_1^\prime}
    - \mathrm{MP}_{CC^\prime \to B_1 B_1^\prime}\psi_{A A^\prime C C^\prime}
  \bigr\|_1
  \le \frac{2 d_B^2}{m},
\end{equation}
where $C$ and $C^\prime$ are registers isomorphic to $B$ and $B^\prime$. The channel $\mathrm{MP}_{CC' \to B_1 B_1'}$ has the following explicit Schur--Weyl (SW) form: it performs a covariant positive operator-valued measure (POVM) on $CC'$ defined by rotating a lowest-weight vector in the irreducible representation (irrep) $\yd{m}$ under the representation $U^{\yd{m}}$. With normalized Haar measure, in the computational basis this is given by $M_U=d_{CC'}^{\yd{m}}U\,\bigl|1,\ldots,1\bigr\rangle
 \bigl\langle1,\ldots,1\bigr|\,U^{\dagger},$ where $d_{CC'}^{\yd{m}}=\multiset{d_B^2}{m}$ is the dimension of the irrep $\mathbb{W}_{CC'}^{\yd{m}}$,
and, for each POVM outcome $U$, prepares on $B_1B_1'$ the corresponding rotated lowest-weight vector $\ket{1}$ in the fundamental irrep.
 Define
\begin{equation}
  \widetilde{\psi}_{A A' B_1 B_1'}
  =\mathrm{MP}_{C C' \to B_1 B_1'}\psi_{A A' C C'}.
\end{equation}
Let us define the approximating marginal by tracing out the appropriate ancillae,
\begin{equation}
  \frac{\widetilde{T}_{A B_1}}{d_A}
  = \Tr_{A' B'_1} \widetilde{\psi}_{A A' B_1 B_1'}.
\end{equation}
By contractivity of the trace norm under completely positive trace-preserving (CPTP) maps, we have
\begin{equation}
\begin{aligned}
  \frac{1}{d_A}\bigl\| T_{A B_1} - \widetilde{T}_{A B_1} \bigr\|_1
  &= \bigl\|\Tr_{A' B_1'} \psi_{A A' B_1 B_1'}
       - \Tr_{A' B_1'} \widetilde{\psi}_{A A' B_1 B_1'}\bigr\|_1 \\
  &\le \bigl\|
       \psi_{A A' B_1 B_1'}
       -  \mathrm{MP}_{C C' \to B_1 B_1'}\psi_{A A' C C'}
     \bigr\|_1 \\
  &\le \frac{2d_B^2}{m}
\end{aligned}
\end{equation}
using~\cref{eq:CKMR_main} in the last step. If the extension is supported on \(\mathbb{W}_B^{\yd{m}}\), the improved estimate follows in the same way from the bosonic de Finetti bound of Ref.~\cite[Theorem~II.8']{CKMR07S}.
\end{proof}

\subsection{EB limit}
\label{subsec:eb_limit_supp}

We can now link the one-site risks to the EB behavior of the Choi matrices established in~\cref{thm:approx_sep_extendible}, and we show that the optimal $1\rightarrow m$ protocols together with their risks converge to an EB channel in the limit $m\rightarrow \infty$.

\begin{theorem}[EB limit of $m$-output CQI protocols]
\label{thm:limit_of_optimal}
For each $m\ge 1$, let $T^{(m)}_{AB_{1,\ldots, m}}$ be the Choi matrix of a $1\to m$ CQI protocol that is exchange-invariant or bosonic on its $m$ outputs. Assume that the marginal Choi matrices $T^{(m)}_{AB_1}$ converge to a Choi matrix $\widetilde{T}_{AB_1}$ in trace norm, and let $\widetilde{\mathcal T}_{AB_1}$ denote the corresponding channel. If the loss $L$ is well behaved, then
\begin{enumerate}
    \item $\widetilde{\mathcal{T}}_{AB_1}$ is an EB channel,
    \item $\widetilde{\mathcal{T}}_{AB_1}$ attains the limiting risk, namely,
\end{enumerate}
\begin{equation}
    \mathcal{L}_{\mleft(\cdot\mright)}\mleft(\widetilde{\mathcal{T}}_{AB_1}\mright)= \lim_{m\to\infty}\mathcal{L}_{\mleft(\cdot\mright),\mathrm{one}}\mleft(\mathcal{T}_{AB}^{(m)}\mright) \, .
\end{equation}
\end{theorem}

\begin{proof}
\cref{thm:approx_sep_extendible} implies that there exists a sequence of EB channels $\widetilde{\mathcal{S}}^{(m)}_{AB_1}$ with Choi matrices $\widetilde{S}^{(m)}_{AB_1}$ such that $\|T_{AB_1}^{(m)}-\widetilde{S}^{(m)}_{AB_1}\|_1\le \varepsilon_m$ with $\varepsilon_m\to 0$ as $m\to\infty$.
By the triangle inequality,
\begin{equation}
    \|\widetilde{S}^{(m)}_{AB_1}-\widetilde{T}_{AB_1}\|_1
    \le
    \|\widetilde{S}^{(m)}_{AB_1}-T^{(m)}_{AB_1}\|_1+
    \|T^{(m)}_{AB_1}-\widetilde{T}_{AB_1}\|_1
    .
\end{equation}
Since both terms on the right-hand side vanish as $m\to\infty$, the left-hand side also converges to zero, and hence $\widetilde{S}^{(m)}_{AB_1}$ converges to $\widetilde{T}_{AB_1}$. By the closedness of the set of EB channel Choi matrices, it follows that $\widetilde{T}_{AB_1}$ is also EB.
By assumed continuity of the risk,
\begin{equation}
    \mathcal{L}_{\mleft(\cdot\mright)}\mleft(\widetilde{\mathcal{T}}_{AB_1}\mright)= \mathcal{L}_{\mleft(\cdot\mright)}\mleft(\lim_{m\rightarrow \infty}\mathcal{T}^{(m)}_{AB_1}\mright) = \lim_{m\rightarrow \infty} \mathcal{L}_{\mleft(\cdot\mright)}\mleft(\mathcal{T}^{(m)}_{AB_1}\mright) = \lim_{m\to\infty}\mathcal{L}_{\mleft(\cdot\mright),\mathrm{one}}\mleft(\mathcal{T}_{AB}^{(m)}\mright) \, .
\end{equation}
\end{proof}

\begin{theorem}[EB limit of minimal risk]
\label{thm:EB_risk}
Let the loss function $L\mleft(\cdot,\cdot\mright)$ be Lipschitz in the second argument with constant $L_{0}$, and let $\mathcal{L}_{\mleft(\cdot\mright),\mathrm{one}}^{\ast}$ and $\widetilde{\mathcal{L}}_{\mleft(\cdot\mright),\mathrm{one}}^{\ast}$ be the optimal coherent and EB risks respectively. Then we have
\begin{equation}
\mathcal{L}_{\mleft(\cdot\mright),\mathrm{one}}^{\ast}
\;\le\;
{\widetilde{\mathcal{L}}_{\mleft(\cdot\mright)}}^{\ast}
\;\le\;
\mathcal{L}_{\mleft(\cdot\mright),\mathrm{one}}^{\ast}+L_0\,\varepsilon_m,
\end{equation}
where $\varepsilon_m=\frac{2d_A d_B^2}{m}$ for the general exchange-invariant case and $\varepsilon_m=\frac{2d_A d_B}{m}$ for the case with bosonic outputs on $\mathbb{W}^{\yd{m}}$.
\end{theorem}

\begin{proof}
Let $T^\ast_{AB}$ be the Choi matrix of an optimal $m$-output symmetric protocol $\mathcal T^\ast_{AB}$ with one-site loss, and let $\widetilde{T}^\ast_{AB_1}$ be the Choi matrix of an optimal EB protocol $\widetilde{\mathcal T}^\ast_{AB_1}$. Since $\widetilde T^\ast_{AB_1}$ is separable, for every $m$ it admits an exchange-invariant extension obtained by replacing each $B$-side factor in a separable decomposition by its $m$-fold tensor power. By refining the decomposition so that the $B$-side factors are pure, the extension can moreover be chosen bosonic. The resulting extension $\widetilde R^{(m)}_{AB}$ has the same $A$-marginal normalization and therefore is the Choi matrix of a channel $\widetilde{\mathcal R}^{(m)}_{AB}$.
\begin{equation}
\Tr_{B_2,\ldots,B_m}\widetilde{R}^{(m)}_{AB}=\widetilde T^\ast_{AB_1},
\end{equation}
where $\mathrm{exc}_B\mleft(\widetilde{R}^{(m)}_{AB}\mright)=\widetilde{R}^{(m)}_{AB}$ or $\widetilde{R}^{(m)}_{AB}\in \mathcal B(\mathbb{H}_A)\otimes \mathcal B(\mathbb{W}_B^{\yd m})$, for the respective exchange-invariant or bosonic case. By the definition of the extension and the one-site loss, we have
\begin{equation}
    {\widetilde{\mathcal{L}}_{\mleft(\cdot\mright)}}^{\ast} = \mathcal{L}_{\mleft(\cdot\mright)}\mleft(\widetilde{\mathcal T}^\ast_{AB_1}\mright) = \mathcal{L}_{\mleft(\cdot\mright),\mathrm{one}}\mleft(\widetilde{\mathcal R}^{(m)}_{AB}\mright) \geq \mathcal{L}_{\mleft(\cdot\mright),\mathrm{one}}^{\ast} \, .
\end{equation}
On the other hand, by~\cref{thm:approx_sep_extendible}, there exists an EB-channel Choi matrix $\widetilde S_{AB_1}$ such that $\|T_{AB_1}^\ast-\widetilde S_{AB_1}\|_1\le \varepsilon_m$. Let $\widetilde{\mathcal S}_{AB_1}$ denote the corresponding EB channel. By optimality of the EB risk and continuity from~\cref{lem:risk_continuity}, we have
\begin{equation}
\widetilde{\mathcal L}_{\mleft(\cdot\mright)}^\ast
\le
\mathcal L_{\mleft(\cdot\mright)}\mleft(\widetilde{\mathcal S}_{AB_1}\mright)
\le
\mathcal L_{\mleft(\cdot\mright)}\mleft(\mathcal T^\ast_{AB_1}\mright)
+L_0\,\varepsilon_m
=
\mathcal L_{\mleft(\cdot\mright),\mathrm{one}}^\ast
+L_0\,\varepsilon_m \, .
\end{equation}
Combining with the first inequality concludes the proof.
\end{proof}

\begin{corollary}[EB limit of optimal protocols]
\label{cor:limit_of_optimal_symmetric}
For each $m\ge 1$, let $T^{(m)\,\ast}_{AB}$ be the Choi matrix of an optimal symmetric $1\to m$ CQI protocol. Assume that the one-site marginals $T^{(m)\,\ast}_{AB_1}$ converge in trace norm to some $T^\ast_{AB_1}$. Then $T^\ast_{AB_1}$ is separable. Moreover, it attains the optimal EB risk
$\mathcal L_{(\cdot)}\mleft(\mathcal T^\ast_{AB_1}\mright)=\widetilde{\mathcal L}_{\mleft(\cdot\mright)}^\ast.$
\end{corollary}

\begin{proof}
We combine~\cref{thm:EB_risk,thm:limit_of_optimal} and use the fact that $\varepsilon_m\to 0$ as $m\to\infty$.
\end{proof}

\section{RP and Approximate Cloning}
\label{sec:rp_cloning_supp}

\subsection{\texorpdfstring{Coherent RP Protocol}{Coherent RP Protocol}}
\label{subsec:coherent_rp_protocol}

RP was first described in Ref.~\cite{TWZ25S} for the case where the number of outputs equals the number of inputs, and has since been studied further in Refs.~\cite{PSTW25S,WW25S,GML26S}. We define it here as a CQI task in the general setting with an arbitrary number of outputs $m$. The task is to take \(n\) copies of a state \(\rho\) satisfying \(\operatorname{rank}(\rho)\leq r\) on input registers \(\mathrm{in}_{1,\ldots,n}\), where each input register corresponds to the Hilbert space \(\mathbb{H}_{\mathrm{in}}\coloneqq \mathbb{C}^d\), and to produce an output state on joint system and environment registers \(\mathrm{out}_{1,\ldots,m}\) and \(\mathrm{env}_{1,\ldots,m}\), where \(\mathbb{H}_{\mathrm{out}}\coloneqq \mathbb{C}^d\), \(\mathbb{H}_{\mathrm{env}}\coloneqq \mathbb{C}^r\), and \(\mathbb{H}_{\mathrm{pur}}\coloneqq \mathbb{H}_{\mathrm{out}}\otimes\mathbb{H}_{\mathrm{env}}\) for each output copy \(i\). We further restrict the output to the totally symmetric subspace \(\mathbb{W}_{\mathrm{pur}_{1,\ldots,m}}^{\yd{m}}\subseteq \mathbb{H}_{\mathrm{pur}_{1,\ldots,m}}\), so that the protocol aims to prepare \(m\) copies of an arbitrary purification of \(\rho\). Here the admissible family \(\mathcal{X}\) is chosen to be all states $\rho\in\mathcal{B}(\mathbb{H}_{\mathrm{in}})$ satisfying $\operatorname{rank}(\rho)\leq r$.
For $n,m>0$ and $\rho\in \mathcal{X}$, let
\begin{equation}
    \mathcal C(\rho) \coloneqq \mathrm{conv}\left\{ \Psi^{\otimes m} \mid \ket{\Psi}\in \mathbb{H}_{\mathrm{pur}} , \, \Tr_{\mathrm{env}}(\Psi) = \rho \right\} \, ,
\end{equation}
That is, \(\mathcal C(\rho)\) is the convex hull of all \(m\)-tensor powers of purifications of \(\rho\). The all-site and one-site RP loss functions then follow from the definitions in~\cref{subsec:one_all_site_risks}.

A natural protocol for this CQI task is given by the \emph{purify-and-clone} channel
\begin{equation}
\mathcal{P}_{n\rightarrow m}^{d,r}
=
\left\{
\begin{aligned}
&\Tr_{m+1,\ldots,n} \circ \mathcal{P}_{n\rightarrow n}^{d,r},
&& m<n, \\
&\mathcal{C}_{n\rightarrow m}^{dr} \circ \mathcal{P}_{n\rightarrow n}^{d,r},
&& m\geq n.
\end{aligned}
\right.
\end{equation}
where $\mathcal{P}_{n\rightarrow n}^{d,r}$ is the RP channel for dimension $d$ and states of rank at most $r$ defined in Ref.~\cite[Section~2.3.3]{TWZ25S}, and $\mathcal{C}_{n\rightarrow m}^{dr}$ is the optimal pure state cloning channel from $n$ to $m$ copies of $\mathbb{C}^{dr}$ described in Ref.~\cite{W98S}. This channel is optimal for $m\leq n$ (see Ref.~\cite{TWZ25S}), and we conjecture that it is also optimal for $m>n$. For infidelity loss, the RP loss is upper bounded by the corresponding $dr$-dimensional pure-state cloning loss, as established in~\cref{prop:random_purification_cost_functions}.

\begin{theorem}
  \label{prop:random_purification_cost_functions}
  Let $\mathcal{X}\subseteq \mathcal{B}(\mathbb{H}_{\mathrm{in}})$ be the set of states $\rho$ satisfying $\operatorname{rank}(\rho)\leq r$. For infidelity loss, the optimal RP risks satisfy
  \begin{subequations}
  \begin{align}
    \mathcal{L}_{\mathrm{RP},\mathrm{w},\mathrm{all}}^{\ast}
    &=
    0,
    && n\geq m, \\
    \mathcal{L}_{\mathrm{RP},\mathrm{w},\mathrm{all}}^{\ast}
    &\leq
    1-\dfrac{\multiset{dr}{n}}{\multiset{dr}{m}},
    && n<m, \\
    \mathcal{L}_{\mathrm{RP},\mathrm{w},\mathrm{one}}^{\ast}
    &=
    0,
    && n\geq m, \\
    \mathcal{L}_{\mathrm{RP},\mathrm{w},\mathrm{one}}^{\ast}
    &\leq
    1-\dfrac{n(m+dr)+m-n}{m(n+dr)}
    =
    \dfrac{(dr-1)(m-n)}{m(n+dr)},
    && n<m .
  \end{align}
  \end{subequations}
\end{theorem}

\begin{proof}
    From Refs.~\cite{W98S,KW99S} we know the exact all-site and one-site fidelities for pure-state cloning with dimension \(dr\) to be
    \begin{subequations}
    \begin{align}
        \mathcal{F}_{\mathrm{cl,all}}^{dr,\ast}
        &= \dfrac{\multiset{dr}{n}}{\multiset{dr}{m}},\\
        \mathcal{F}_{\mathrm{cl,one}}^{dr,\ast}
        &= \dfrac{n(m+dr)+m-n}{m(n+dr)} \, .
    \end{align}
    \end{subequations}
    For any operator $X\in\mathcal B(\mathbb H_{\mathrm{pur}}^{\otimes n})$, define the twirling channel by
    \begin{equation*}
        \tag{\theequation c}
        \operatorname{twirl}(X)
        \coloneqq
        \int_{U(r)}
        (I_{\mathrm{out}}\otimes U)^{\otimes n}
        X
        (I_{\mathrm{out}}\otimes U^\dagger)^{\otimes n}\,\mathrm dU,
    \end{equation*}
    where $\mathrm dU$ denotes normalized Haar measure. Fix a purification $\Psi$ of $\rho$.
    \begin{equation}
    \begin{aligned}
        1 - L_{\mathrm{all}}\mleft(\mathcal C(\rho),\mathcal{P}_{n\rightarrow m}^{d,r}\left(\rho^{\otimes n}\right)\mright) &\geq F\!\left(\mathcal{P}_{m\rightarrow m}^{d,r}(\rho^{\otimes m}),\mathcal{P}_{n\rightarrow m}^{d,r}\left(\rho^{\otimes n}\right)\right) \\
        &= F\!\left(\operatorname{twirl}(\Psi^{\otimes m}),\mathcal{C}_{n\rightarrow m}^{dr}\left(\operatorname{twirl}(\Psi^{\otimes n})\right)\right) \\
        &\geq F\!\left(\Psi^{\otimes m}, \mathcal{C}_{n\rightarrow m}^{dr}\left(\Psi^{\otimes n}\right)\right) \\
        &= \mathcal{F}_{\mathrm{cl,all}}^{dr,\ast} \, .
    \end{aligned}
    \end{equation}
    The first step uses the definition of $L_{\mathrm{all}}$, since the random purification $\mathcal{P}_{m\rightarrow m}^{d,r}(\rho^{\otimes m})$ is an element of $\mathcal C(\rho)$. The second equality uses Ref.~\cite{TWZ25S}, which gives $\mathcal{P}_{n\rightarrow n}^{d,r}(\rho^{\otimes n})=\operatorname{twirl}(\Psi^{\otimes n})$. The third inequality follows from the data-processing inequality for fidelity and the unitary covariance of $\mathcal C_{n\rightarrow m}^{dr}$. The same argument applied to the one-site loss gives the one-site risk.
\end{proof}

\subsection{\texorpdfstring{Incoherent Benchmark}{Incoherent Benchmark}}
\label{subsec:rp_eb_benchmark}

To establish this incoherent bound, we relate RP to the corresponding mixed-state cloning task. For cloning, the input and output Hilbert spaces are $\mathbb{H}_\mathrm{in}=\mathbb{H}_\mathrm{out}=\mathbb{C}^{d}$, and the target set is the singleton $\mathcal C_{\mathrm{cl}}(\rho)=\{\rho^{\otimes m}\}$, with the all-site and one-site loss functions defined as in~\cref{subsec:one_all_site_risks}.
In general, for \(\Gamma\in \mathcal C(\rho)\) we have
\(\Tr_{\mathrm{env}_{1,\ldots,m}}(\Gamma)=\rho^{\otimes m}\), and in the one-site case taking the maximum over \(i\) preserves the ordering. Hence,
for \((\cdot)\in\{\mathrm{all},\mathrm{one}\}\),
\begin{equation}
    L_{(\cdot)}\mleft(\mathcal C(\rho),\Xi\mright)
    \geq
    L_{(\cdot)}\mleft(\mathcal C_{\mathrm{cl}}(\rho),\Tr_{\mathrm{env}_{1,\ldots,m}}(\Xi)\mright) \, .
\end{equation}
Hence, tracing out the environment of any RP protocol \(\mathcal T\) yields a mixed-state cloning protocol \(\operatorname{Tr}_{\mathrm{env}}\circ\mathcal T\). Since this reduction preserves incoherence, any lower bound on the incoherent cloning risk also applies to the incoherent RP risk:
\begin{equation}
    \label{equ:rp_cloning_risk_ordering}
    \mathcal{L}_{\mathrm{RP},\mathrm{w},(\cdot)}^{\ast}
    \geq
    \mathcal{L}_{\mathrm{cl},\mathrm{w},(\cdot)}^{\ast} \, .
\end{equation}

It therefore remains only to evaluate the cloning risks on the right-hand side. For this purpose, it suffices to restrict attention to the pure-state subfamily, which yields a lower bound on the RP risk over states of rank at most $r$. On pure states, the fidelity is linear in the prepared output state, so the optimal pure-state estimation result of Ref.~\cite{H13S} gives the incoherent risk bound.

For \(d\)-dimensional pure states, the optimal all-site and one-site worst-case risks of any incoherent protocol are given in Ref.~\cite{H13S}
\begin{subequations}
\begin{align}
\widetilde{\mathcal{L}}_{\mathrm{cl},\mathrm{w},\mathrm{all}}^{\ast}
&=
1-\frac{\multiset{d}{n}}
{\multiset{d}{n+m}},
\\
\widetilde{\mathcal{L}}_{\mathrm{cl},\mathrm{w},\mathrm{one}}^{\ast}
&=\frac{d-1}{n+d}.
\end{align}
\end{subequations}

In the fixed-excess regime \(m=n+\ell\), the contrast between coherent RP
and incoherent RP is particularly transparent. For the
coherent RP protocol of~\cref{prop:random_purification_cost_functions},
\begin{equation}
\mathcal{L}_{\mathrm{RP},\mathrm{w},\mathrm{all}}^{\ast}
\leq
1-\frac{\multiset{dr}{n}}{\multiset{dr}{n+\ell}}
=
\frac{\ell(dr-1)}{n}
+
O\!\left(\frac{1}{n^2}\right),
\end{equation}
whereas the optimal incoherent RP risk obeys
\begin{equation}
\widetilde{\mathcal{L}}_{\mathrm{RP},\mathrm{w},\mathrm{all}}^{\ast}
\geq
1-\frac{\multiset{d}{n}}{\multiset{d}{2n+\ell}}
=
1-2^{1-d}
+
O\!\left(\frac1n\right).
\end{equation}
Thus the coherent RP all-site error vanishes as \(n\to\infty\), while the
incoherent RP all-site error remains bounded away from zero. For the one-site marginal figure of merit, set $m=n+\ell$. The exact one-site infidelity of the $dr$-dimensional universal $n\to m$ cloner is $\ell(dr-1)/((n+\ell)(n+dr))$, which bounds the coherent RP loss of~\cref{prop:random_purification_cost_functions}.
Using $(n+\ell)(n+dr)\geq n^2$ and $dr-1\leq dr$ gives the uniform bound
\begin{equation}
\mathcal L_{\mathrm{RP},\mathrm{w},\mathrm{one}}^\ast
\leq
\frac{\ell dr}{n^2}
.
\end{equation}
Hence choosing $n\geq\sqrt{\ell dr/\varepsilon}$ is sufficient to guarantee $\mathcal L_{\mathrm{RP},\mathrm{w},\mathrm{one}}^\ast\leq\varepsilon$. Writing $n_{\mathrm{coh}}(\varepsilon)$ for the corresponding coherent sample complexity, we obtain $n_{\mathrm{coh}}(\varepsilon)\leq\left\lceil\sqrt{\ell dr/\varepsilon}\right\rceil$, and consequently $n_{\mathrm{coh}}=\bigO{\sqrt{\ell dr/\varepsilon}}$.
By contrast, the incoherent RP risk is lower bounded by pure-state estimation:
\begin{equation}
\widetilde{\mathcal{L}}_{\mathrm{RP},\mathrm{w},\mathrm{one}}^{\ast}
\geq
\frac{d-1}{n+d}.
\end{equation}
Writing $n_{\mathrm{inc}}(\varepsilon)$ for the corresponding incoherent sample complexity, achieving risk at most $\varepsilon$ requires $n_{\mathrm{inc}}(\varepsilon)\geq(d-1)/\varepsilon-d$. For $d\geq2$ and $\varepsilon\leq1/4$, we have $d-1\geq d/2$ and $d\leq d/(4\varepsilon)$. Hence $n_{\mathrm{inc}}(\varepsilon)\geq d/(4\varepsilon)$, yielding the uniform bound $n_{\mathrm{inc}}=\bigOm{d/\varepsilon}$.

As an aside, the reduction in~\cref{equ:rp_cloning_risk_ordering}, together with the bounds above, shows that the coherent--incoherent separation also holds for mixed-state cloning.

\section{QPA and Eigenstate Tomography}
\label{sec:qpa_supp}

For QPA, we fix a sorted spectrum \(\boldsymbol{p}=(p_{1},\ldots,p_{d})\), with \(p_{1}\geq \cdots \geq p_{d}\), and take the admissible input family to be its unitary orbit:
\begin{equation}
    \mathcal{X}\coloneqq \left\{U\mathrm{diag}(\boldsymbol{p})U^\dagger \, \mid \, U\in SU(d)\right\} \, .
\end{equation}
Throughout this section, we use infidelity as the loss, \(L(\rho,\sigma)\coloneqq 1-F(\rho,\sigma)\), and define the singleton target set by \(\mathcal C(\rho)\coloneqq\{{^k}\gamma(\rho)\}\), where \({^k}\gamma(\rho)\) denotes the projector onto the \(k\)th eigenvector of \(\rho\).
For further details of the setup, see Definition~1 (General QPA) of Ref.~\cite{QPA26S}.

\subsection{Nonasymptotic QPA sample-complexity bounds}
\label{subsec:qpa_sample_complexity_supp}

\begin{theorem}[Explicit upper bound on the sample complexity]
\label{cor:explicit_sample_complexity_inversion}
Assume \(0<D_{k,\mathrm{min}}\le 1\), \(m\ge 1\), and let \(\varepsilon\) denote the target all-site QPA risk, with \(0<\varepsilon\le \min\{1,(2D_{k,\mathrm{min}})^{-1}\}\).
The minimal sample complexity satisfies
\begin{equation}
n^\ast
\le
\frac{m}{\varepsilon D_{k,\mathrm{min}}^2}
+
R.
\end{equation}
Here \(S\coloneqq \frac{\sqrt m}{\sqrt\varepsilon D_{k,\mathrm{min}}^2}\), and the remainder satisfies \(R<CS\ln(CS)\) for all \(S>S_0\), with \(C\) and \(S_0\) constants.
\end{theorem}

\begin{proof}
Set $N\coloneqq \frac{m}{\varepsilon D_{k,\mathrm{min}}^2}.$
We show that there exist universal constants \(C>0\) and \(S_0>0\) such that for all \(S\ge S_0\), taking
\begin{equation}
n\ge N+CS\ln(CS)
\end{equation}
ensures \({^k}\mathcal L_{\mathrm{w},\mathrm{all}}^\ast\le \varepsilon\).
We choose \(S_0>2\) such that \(S^2\geq CS\ln(CS)\) for all \(S\ge S_0\), with \(C=8c^2\) and $c=12$.

(I) Since \(m\ge1\), \(0<\varepsilon\le1\), and \(D_{k,\mathrm{min}}\le1\), we have
$S=\frac{\sqrt m}{D_{k,\mathrm{min}}^2\sqrt\varepsilon}
\ge
\frac{1}{D_{k,\mathrm{min}}^2}.$
The concentration threshold in Ref.~\cite[Theorem~S5.37]{QPA26S} is satisfied since
\begin{subequations}
\begin{align}
{\color{revisionblue}n}&{\color{revisionblue}\ge CS\ln(CS)
\ge C D_{k,\mathrm{min}}^{-2}\ln(CD_{k,\mathrm{min}}^{-2}),}\label{eq:qpa_inversion_concentration_threshold}\\
{\color{revisionblue}n}&{\color{revisionblue}\ge N
=\frac{m}{\varepsilon D_{k,\mathrm{min}}^2}
\ge\frac{2m}{D_{k,\mathrm{min}}}
\ge\frac{2m}{D_{k,k+1}}.}\label{eq:qpa_inversion_overhang_threshold}
\end{align}
\end{subequations}
In Eq.~\eqref{eq:qpa_inversion_overhang_threshold}, we use the convention \(D_{d,d+1}\coloneqq\infty\).

(II) Using \(1 \ge |D_{k,i}|\ge D_{k,\mathrm{min}}\ge 0\), \(\sum_{i\neq k}p_i\le1\), we have \(4\left(1+\frac{6m}{D_{k,\mathrm{min}}^3}\right)
\le
\frac{28m}{D_{k,\mathrm{min}}^3}\), and the worst-case risk upper bound in Ref.~\cite[Theorem~S5.37]{QPA26S} gives
\begin{equation}
\label{eq:simplified_loss_bound_for_inversion}
{^k}\mathcal L_{\mathrm{w},\mathrm{all}}^\ast
\le
\frac{m}{nD_{k,\mathrm{min}}^2}
+
\frac{28m\sqrt{\ln n}}{n^{3/2}D_{k,\mathrm{min}}^3}.
\end{equation}
Let \(n_0\coloneqq N+CS\ln(CS)\).
We show that the second term in~\cref{eq:simplified_loss_bound_for_inversion} satisfies
\begin{align}
\frac{28m\sqrt{\ln n}}{n^{3/2}D_{k,\mathrm{min}}^3}
\le
\frac{1}{n}
\frac{28m}{D_{k,\mathrm{min}}^3}
\sqrt{\frac{\ln n_0}{n_0}}
\le
\frac{1}{n}
\frac{28m}{D_{k,\mathrm{min}}^3}
\frac{3\ln S}{\sqrt N}
=
\frac{84\varepsilon S\ln S}{n}
\le
\frac{\varepsilon CS\ln(CS)}{n}.
\label{eq:explicit_inversion_tail_bound}
\end{align}
Here, the first inequality follows because \(\sqrt{\ln x/x}\) is decreasing for \(x>e\) and \(n\ge n_0>e\).
In the second step, from our choice of $S$, we have \(CS\ln(CS)\le S^2\) and \(N=m/(\varepsilon D_{k,\mathrm{min}}^2)\le m/(\varepsilon D_{k,\mathrm{min}}^4)=S^2\). Hence \(n_0\le 2S^2\le S^3\), and therefore \(\sqrt{\ln n_0}\le\sqrt{3\ln S}<3\ln S\).
Moreover, since \(n_0\ge N\), we also have \(\sqrt{n_0}\ge \sqrt N\).
In the third step, we use the definition of $N$ and $S$. Finally, the last inequality uses \(C\ge84\) and \(\ln(CS)\ge\ln S\).

On the other hand, since \(n\ge N+CS\ln(CS)\),
\begin{equation}
\label{eq:explicit_inversion_leading_gain}
\begin{aligned}
\varepsilon-\frac{m}{nD_{k,\mathrm{min}}^2}
=
\frac{m}{D_{k,\mathrm{min}}^2}
\left(
\frac{1}{N}-\frac{1}{n}
\right)
=
\frac{m}{D_{k,\mathrm{min}}^2}
\frac{n-N}{Nn}
\ge
\frac{m}{D_{k,\mathrm{min}}^2}
\frac{CS\ln(CS)}{Nn}
=
\frac{\varepsilon CS\ln(CS)}{n}.
\end{aligned}
\end{equation}
Therefore, combining~\cref{eq:explicit_inversion_tail_bound} and~\cref{eq:explicit_inversion_leading_gain} and substituting into~\cref{eq:simplified_loss_bound_for_inversion} gives
\begin{equation}
\label{eq:explicit_inversion_final_loss}
{^k}\mathcal L_{\mathrm{w},\mathrm{all}}^\ast
\le
\frac{m}{nD_{k,\mathrm{min}}^2}
+
\left(
\varepsilon-\frac{m}{nD_{k,\mathrm{min}}^2}
\right)
=
\varepsilon.
\end{equation}
Since \(n\) must be an integer, the above shows
\[
n^\ast
\le
\left\lceil N+CS\ln(CS)\right\rceil.
\]
By slightly increasing the universal constants \(C\) and \(S_0\), the ceiling can be absorbed into the \(CS\ln(CS)\) term.
\end{proof}

\begin{lemma}[Constant typical subspace risk bound]
\label{lem:relative_gap_averaging_bound}
Let \(J=\{1\}\) if \(k=1\), \(J=\{d-1\}\) if \(k=d\), and \(J=\{k-1,k\}\) otherwise, and write \(\overline{\Delta}_{i,j}\coloneqq\Delta_{i,j}/n\). Assume that the sector-wise loss satisfies \(1-f^{\yd{\varsigma}}\leq\sum_{j\in J}\frac{c_j}{n\overline{\Delta}_{j,j+1}}\) for \(n\geq m/\overline{\Delta}_{k,k+1}\) when \(k\neq d\), with all \(c_j\geq0\).
Then for every \(0<\beta<1\) with \(n\geq \frac{m}{(1-\beta)D_{k,k+1}}\) when \(k\neq d\), the overall loss satisfies
\begin{equation}
\label{eq:relative_gap_loss_bound}
    {^k}\mathcal L_{\mathrm{w},\mathrm{all}}^\ast\leq
    \frac{1}{n(1-\beta)}
    \sum_{j\in J}\frac{c_j}{D_{j,j+1}}
    +
    2|J|e^{-\frac{(\beta\sqrt nD_{k,\mathrm{min}}-4)^2}{32}}.
\end{equation}
\end{lemma}

\begin{proof}
Let \(\mathsf T_\beta\coloneqq \{\left|\overline{\Delta}_{j,j+1}-D_{j,j+1}\right|\leq\beta D_{k,\mathrm{min}}\text{ for all }j\in J\}\).
For \(k\neq d\), on \(\mathsf T_\beta\),
\[
\overline{\Delta}_{k,k+1}
\geq D_{k,k+1}-\beta D_{k,\mathrm{min}}
\geq (1-\beta)D_{k,k+1},
\]
and hence
\[
n\geq \frac{m}{(1-\beta)D_{k,k+1}}
\geq \frac{m}{\overline{\Delta}_{k,k+1}},
\]
so the sector-wise bound is applicable on the typical event.
On this event, \(\overline{\Delta}_{j,j+1}\geq(1-\beta)D_{j,j+1}\), and hence the assumed sector-wise bound gives
\(1-f^{\yd{\varsigma}}\leq n^{-1}(1-\beta)^{-1}\sum_{j\in J}c_j/D_{j,j+1}\).
Averaging, using the trivial bound on \(\mathsf T_\beta^c\), and applying the row-difference tail bound with a union bound gives the claimed estimate.
\end{proof}

\begin{corollary}[Constant typical subspace sample complexity upper bound]
\label{cor:clean_dmin_sample_complexity}
Let
\begin{equation}
J=
\left\{
\begin{aligned}
&\{1\}, && k=1,\\
&\{k-1,k\}, && 1<k<d,\\
&\{d-1\}, && k=d,
\end{aligned}
\right.
\end{equation}
Assume \(m\geq1\) and \(0<\varepsilon\leq1\).
If \(k\in\{1,d\}\), then \(n\geq98m/(\varepsilon D_{k,\mathrm{min}}^2)\) suffices to ensure \({^k}\mathcal L_{\mathrm{w},\mathrm{all}}^\ast\leq\varepsilon\). If \(1<k<d\), then \(n\geq135m/(\varepsilon D_{k,\mathrm{min}}^2)\) suffices. Since \(n^\ast\) is an integer, this yields the convenient bounds
\begin{equation}
n^\ast
\leq
\left\{
\begin{aligned}
&\dfrac{99m}{\varepsilon D_{k,\mathrm{min}}^2}, && k\in\{1,d\},\\[0.8em]
&\dfrac{136m}{\varepsilon D_{k,\mathrm{min}}^2}, && 1<k<d.
\end{aligned}
\right.
\end{equation}
\end{corollary}

\begin{proof}
For \(\beta=15/16\) when \(k=1\), and for \(\beta=23/24\) when \(1<k<d\), the respective sample-size bounds imply \(n\geq16m/D_{k,k+1}\) and \(n\geq24m/D_{k,k+1}\), so the sector-wise bound of~Ref.~\cite[Corollary~S5.33]{QPA26S} is applied within its applicability domain. For \(k=d\), no such condition is required.
For \(k\in\{1,d\}\), applying~\cref{lem:relative_gap_averaging_bound} with the corresponding \(J\), \(D_{k,\mathrm{min}}\), and \(\beta=15/16\) gives
\begin{equation}
{^k}\mathcal L_{\mathrm{w},\mathrm{all}}^\ast\leq
\frac{16m}{nD_{k,\mathrm{min}}^2}
+
2e^{-\frac{\mleft(\frac{15}{16}\sqrt nD_{k,\mathrm{min}}-4\mright)^2}{32}},
\end{equation}
where we used \(\sum_{j\in J}c_j/D_{j,j+1}\leq m/D_{k,\mathrm{min}}^2\). If \(n\geq98m/(\varepsilon D_{k,\mathrm{min}}^2)\), the first term is at most \(8\varepsilon/49\). For the exponential term, the quantity $2\varepsilon^{-1}e^{-\mleft(\frac{15}{16}\sqrt{98/\varepsilon}-4\mright)^2/32}$
is increasing on \(0<\varepsilon\leq1\), and at \(\varepsilon=1\) it is smaller than \(41/49\). Hence the exponential term is at most \(41\varepsilon/49\), and \({^k}\mathcal L_{\mathrm{w},\mathrm{all}}^\ast\leq\varepsilon\).

For \(1<k<d\), applying~\cref{lem:relative_gap_averaging_bound} with \(J=\{k-1,k\}\) and \(\beta=23/24\) gives
\begin{equation}
{^k}\mathcal L_{\mathrm{w},\mathrm{all}}^\ast\leq
\frac{24m}{nD_{k,\mathrm{min}}^2}
+
4e^{-\frac{\mleft(\frac{23}{24}\sqrt nD_{k,\mathrm{min}}-4\mright)^2}{32}},
\end{equation}
using \(\sum_{j\in J}c_j/D_{j,j+1}\leq m/D_{k,\mathrm{min}}^2\). If \(n\geq135m/(\varepsilon D_{k,\mathrm{min}}^2)\), the first term is at most \(8\varepsilon/45\). The quantity $4\varepsilon^{-1}e^{-\mleft(\frac{23}{24}\sqrt{135/\varepsilon}-4\mright)^2/32}$
is increasing on \(0<\varepsilon\leq1\), and at \(\varepsilon=1\) it is smaller than \(37/45\). Hence the exponential term is at most \(37\varepsilon/45\), and again \({^k}\mathcal L_{\mathrm{w},\mathrm{all}}^\ast\leq\varepsilon\).
Finally, since \(m/(\varepsilon D_{k,\mathrm{min}}^2)\geq1\), integrality of \(n^\ast\) gives \(n^\ast\leq\left\lceil98m/(\varepsilon D_{k,\mathrm{min}}^2)\right\rceil<99m/(\varepsilon D_{k,\mathrm{min}}^2)\) for \(k\in\{1,d\}\), and \(n^\ast\leq\left\lceil135m/(\varepsilon D_{k,\mathrm{min}}^2)\right\rceil<136m/(\varepsilon D_{k,\mathrm{min}}^2)\) for \(1<k<d\).
The second case is used in the main text.
\end{proof}

\subsection{Eigenstate Tomography}
\label{subsec:eigenstate_tomography_supp}

The natural incoherent counterpart of QPA is a measurement-mediated protocol that first performs $k$-th eigenstate tomography and then prepares the estimated eigenstate.
To make this explicit, we recall the $k$-th eigenstate-tomography task.
Given \(n\) copies of \(\rho\) on \(\mathbb{C}^d\), the goal is to estimate the target eigenstate \(\ket{\psi_k}\). More precisely, we have a measurement \(\int_{X}M_{x} \, \mathrm{d}x = I_{d}^{\otimes n}\) with corresponding estimators \(\sigma_x\in\mathcal{B}(\mathbb{C}^d)\), and some loss function \(L(\cdot,\cdot)\), usually the infidelity or trace distance. Set
\begin{equation}
    {^k}\gamma(\rho)\coloneqq\ketbra{\psi_{k}}{\psi_{k}} \, ,
\end{equation}
where \(\ket{\psi_k}\) is the \(k\)-th eigenstate of \(\rho\). For some set of states \(\mathcal{X}\), the worst-case risk associated to the measurement \(\{M_x\}_{x\in X}\) is now given as
\begin{equation}
    {^k}\widehat{\mathcal{L}}_{\mathrm{w}}(\{M_x\}_{x\in X}) \coloneqq \sup_{\rho\in\mathcal{X}}\int_{X}\Tr\mleft(M_{x}\rho^{\otimes n}\mright)L({^k}\gamma(\rho),\sigma_x) \, \mathrm{d}x \, .
\end{equation}
If there is a probability measure \(\pi\) over \(\mathcal{X}\), we can define the average risk in a similar way
\begin{equation}
    {^k}\widehat{\mathcal{L}}_{\mathrm{a}}(\{M_x\}_{x\in X}) \coloneqq \int_{\mathcal{X}}\int_{X}\Tr\mleft(M_{x}\rho^{\otimes n}\mright)L({^k}\gamma(\rho),\sigma_x) \, \mathrm{d}x \, \mathrm{d}\pi(\rho) \, .
\end{equation}
One question is whether all estimators \(\sigma_x\) have to be physical, i.e., PSD matrices with trace \(1\), or whether they can be any operators. As it turns out, for metric losses, the difference is at most a constant factor.

\begin{lemma}[Physical and unconstrained metric tomography]
\label{lem:physical_unconstrained_norm_tomography}
Assume that the loss is induced by a metric \(d\), namely \(L\mleft(\cdot,\cdot\mright)=d(\cdot,\cdot)\).
Let ${^k}\widehat{\mathcal{L}}_{(\cdot),\mathrm{phys}}^{\ast}$ denote the minimal risk over physical estimators, and let ${^k}\widehat{\mathcal{L}}_{(\cdot)}^{\ast}$ denote the minimal risk over unconstrained (arbitrary matrix-valued) estimators. Then
\begin{equation}
    {^k}\widehat{\mathcal{L}}_{(\cdot)}^{\ast} \leq {^k}\widehat{\mathcal{L}}_{(\cdot),\mathrm{phys}}^{\ast} \leq 2{^k}\widehat{\mathcal{L}}_{(\cdot)}^{\ast} \, .
\end{equation}
\end{lemma}

\begin{proof}
Since every physical estimator is also an unconstrained estimator, one always has
${^k}\widehat{\mathcal{L}}_{(\cdot)}^{\ast} \leq {^k}\widehat{\mathcal{L}}_{(\cdot),\mathrm{phys}}^{\ast}$. For the converse direction, given an arbitrary matrix-valued estimator $\widehat{\varpi}_{x}$, define the physical estimator $\widehat{\sigma}(\widehat{\varpi}_{x})$ as the nearest-state map from matrices to states with respect to the metric. Such a minimizer exists by compactness and continuity, so the relevant expectation is well defined. Then, for every input state $\rho\in\mathcal{X}$, we have
\begin{equation}
\begin{aligned}
d(\widehat{\sigma}(\widehat{\varpi}_{x}),{^k}\gamma(\rho)) \le d(\widehat{\sigma}(\widehat{\varpi}_{x}),\widehat{\varpi}_{x})+d(\widehat{\varpi}_{x},{^k}\gamma(\rho)) \le 2d(\widehat{\varpi}_{x},{^k}\gamma(\rho)).
\end{aligned}
\end{equation}
This gives the result.
\end{proof}

This result holds in particular for the trace distance, and it is also compatible with stronger tomography guarantees, such as high-probability estimators that are close under the relevant metric~\cite{HCET+24S}.

For infidelity loss, the minimal risk of $k$-th eigenstate tomography with physical estimators is actually equal to that of the corresponding incoherent protocol. Thus, the results obtained through the EB limit also apply directly to $k$-th eigenstate tomography. We have
\begin{equation}
    L(\rho,\sigma) = 1 - F(\rho,\sigma) \, .
\end{equation}
Since the target state is pure, we can write the fidelity in terms of the trace and get
\begin{equation}
    \int_{X}\Tr\mleft(M_{x}\rho^{\otimes n}\mright)\left(1-\Tr\mleft({^k}\gamma(\rho)\sigma_x\mright)\right) \, \mathrm{d}x = 1 - \Tr\mleft({^k}\gamma(\rho)\int_{X}\sigma_x\Tr\mleft(M_{x}\rho^{\otimes n}\mright) \, \mathrm{d}x\mright) = L({^k}\gamma(\rho),\widetilde{\mathcal{T}}_{M,\sigma}(\rho^{\otimes n})) \, ,
\end{equation}
where \(\widetilde{\mathcal{T}}_{M,\sigma}\) is the incoherent channel given by the measurement \(\{M_x\}_{x\in X}\) and the prepared states \(\{\sigma_x\}_{x\in X}\). This means that the risks for tomography with physical estimators and EB channels are equivalent, and therefore
\begin{equation}
    {^k}\widehat{\mathcal{L}}_{(\cdot),\mathrm{phys}}^{\ast} = {^k}\widetilde{\mathcal{L}}_{(\cdot)}^{\ast} \, .
\end{equation}
We finally remark that our result also translates to other figures of merit for the $k$-th eigenstate tomography task, for instance high-probability guarantees~\cite{HCET+24S}.

\subsection{POVM of the EB Channel}
\label{subsec:eb_povm_supp}
We begin by considering the one-site marginal of the identity channel on the totally symmetric subspace \(\mathbb{W}^{\yt{1\cdots m}}\). For large \(m\), this marginal is asymptotically EB and admits a covariant measurement-mediated representation where \(\mathrm{d}\psi\) is Haar measure on pure states, as in Ref.~\cite{H13S}:
\begin{equation}
\begin{aligned}
\Tr_{2,\ldots,m}(\mathcal I^{\yt{1\cdots m}}(\cdot))
=&\int \Tr\left(\multiset{d}{m}\,\psi^{\otimes m}(\cdot)\right)\,\psi\,\mathrm{d}\psi
+\bigO{m^{-1}},
\end{aligned}
\end{equation}

For each sector, recall that the sector-wise QPA channel based on the overhang-removal rule can be written in the Stinespring form
\begin{equation}
\extChannel{\mathcal{T}}{\yd{\varsigma}}{\yt{1\cdots m}}{\yd{\mu}}(\cdot)
=
\Tr_{\yd{\mu}}\!\left(\intertwiner{W}{\yd{\varsigma}}{\yd{\mu}}{\yt{1\cdots m}}{}\,\cdot\,{\intertwineradj{W}{\yd{\varsigma}}{\yd{\mu}}{\yt{1\cdots m}}{}}\right).
\end{equation}
To study the limiting form of this channel as \(m\to\infty\), we evaluate the pulled-back POVM
\begin{equation}
  \label{eq:pulled_back_povm}
M_I={\extChannel{\mathcal{T}}{\yd{\varsigma}}{\yt{1\cdots m}}{\yd{\mu}}}^\star\left(\multiset{d}{m}\psi^{\otimes m}\right)=\multiset{d}{m}\intertwineradj{W}{\yd{\varsigma}}{\yd{\mu}}{\yt{1\cdots m}}{}I_\yd{\mu}\otimes\psi^{\otimes{m}}{\intertwiner{W}{\yd{\varsigma}}{\yd{\mu}}{\yt{1\cdots m}}{}}.
\end{equation}
We show that it has a well-defined limit. This yields convergence of the marginal channel and identifies the limiting map as EB, with an explicit POVM and corresponding re-preparation operators.

Since the measurement is unitary invariant, it suffices to understand $M_I$ for $\psi = \ketbra{d}{d}$. Consider the action of $U\in U(d-1)$, embedded in $U(d)$ via the standard inclusion $U(d-1)\leq U(d)$ and defined on the fundamental representation by the block-diagonal matrix $U\oplus [1]$. The pulled-back POVM~\cref{eq:pulled_back_povm} is invariant under this action. By Schur's lemma, it therefore decomposes as $M_I=\sum_{\yd{\alpha}\,\vdash\,\left.\yd{\varsigma}\right|_{d-1}^{d}} c_m^{\yd{\alpha}}\,I^{\yd{\alpha}},$
where \(\yd{\alpha}\) runs over all Young diagrams (YDs) obtained from Young tableaux (YTs) of \(\yd{\varsigma}\) by removing all boxes labeled $d$.
To evaluate the constants $c_m^{\yd{\alpha}}$, fix $\yd{\alpha}\vdash\left.\yd{\varsigma}\right|_{d-1}^{d}$ and choose a Gel'fand--Tsetlin (GT) basis vector $\wt{w}$ whose first row is $\yd{\varsigma}$ and second row is $\yd{\alpha}$. Then $\langle M_I\rangle_{\ket{\wt{w}}}=c_m^{\yd{\alpha}}$, which by the Clebsch--Gordan coefficient (CGC) fidelity component formula evaluates to
\begin{equation}
\label{eq:constant_eval}
c_m^{\yd{\alpha}}
=
\res{f}{\yd{\mu}}{\yd{m}}{\yd{\varsigma}}_{\mathrm{all}}(\wt{w})
=
\sum_{\wt{v}\vdash\yd{\mu}}
\mathrm{CGC}^2\!\left(\wt{v},\,\wt{d,\ldots,d}\mid\wt{w}\right).
\end{equation}

To evaluate the constant explicitly, we consider the case where $\yd{\mu}$ is defined by overhang-removal for the $k$-th target eigenstate and has terminal row \(d\). We have
\begin{equation}
m_\ell=\left\{\begin{aligned}
&0, &&\ell<k,\\
&\Delta_{\ell,\ell+1}, &&k\le \ell\le d-1,\\
&m-\Delta_{k,d}, &&\ell=d.
\end{aligned}\right.
\end{equation}
Moreover, we pick $\yd{\alpha}$ to be
\begin{equation}
\begin{aligned}
\alpha_i=
\left\{\begin{aligned}
&\varsigma_i, && i<k,\\
&\varsigma_{i+1}, && i\ge k.
\end{aligned}\right.
\end{aligned}
\end{equation}
which forces the edge variables $t_{i,j}$ to be $0$ in Ref.~\cite[Eq.~(69)]{QPA26S}. Therefore, one gets the following expression. Colors represent corresponding factors through the simplification.
\begin{equation}
\begin{aligned}
&\res{f}{\yd{\mu}}{\yd{m}}{\yd{\varsigma}}_\mathrm{all}(\wt{w})
=\binom{m}{\mathbf m}
\prod_{\ell=k}^{d}
{\color{defgreen}\prod_{i=1}^{k-1}
\prod_{r=1}^{m_\ell}
\left(1-\frac{1}{\Delta_{i,\ell}-i+\ell+r}\right)}
\prod_{\ell=k}^{d}
{\color{jam}\frac{\rpoch{1}{m_\ell}}{\rpoch{\Delta_{k,\ell}-k+\ell+1}{m_\ell}}}.
\end{aligned}
\end{equation}
The \(r\) product telescopes:
\begin{equation}
\begin{aligned}
{\color{defgreen}\prod_{r=1}^{m_\ell}
\left(1-\frac{1}{\Delta_{i,\ell}-i+\ell+r}\right)
=
\prod_{r=1}^{m_\ell}
\frac{\Delta_{i,\ell}-i+\ell+r-1}{\Delta_{i,\ell}-i+\ell+r}
=
\frac{\Delta_{i,\ell}-i+\ell}{\Delta_{i,\ell}-i+\ell+m_\ell}}.
\end{aligned}
\end{equation}
Since \(\rpoch{1}{m_\ell}=m_\ell!\), the multinomial prefactor simplifies to $\binom{m}{\mathbf m}\prod_{\ell=k}^{d}\rpoch{1}{m_\ell}
=m!.$
Hence
\begin{equation}
\label{eq:f_reduced_app}
\begin{aligned}
\res{f}{\yd{\mu}}{\yd{m}}{\yd{\varsigma}}_\mathrm{all}(\wt{w})
=
m!
\prod_{\ell=k}^{d}{\color{defgreen}\prod_{i=1}^{k-1}
\frac{\Delta_{i,\ell}-i+\ell}{\Delta_{i,\ell}-i+\ell+m_\ell}}
{\color{jam}\frac{1}{\rpoch{\Delta_{k,\ell}-k+\ell+1}{m_\ell}}}.
\end{aligned}
\end{equation}
We substitute the row-differences of the overhang-removal configuration into \(m_\ell\). For each fixed \(i<k\), telescoping yields
\begin{equation}
\begin{aligned}
{\color{defgreen}\prod_{\ell=k}^{d}
\frac{\Delta_{i,\ell}-i+\ell}{\Delta_{i,\ell}-i+\ell+m_\ell}
=
\frac{\Delta_{i,k}+k-i}{\Delta_{i,k}-i+d+m}
\prod_{\ell=k+1}^{d}
\frac{\Delta_{i,\ell}-i+\ell}{\Delta_{i,\ell}-i+\ell-1}}.
\end{aligned}
\end{equation}
Moreover,
\begin{equation}
\begin{aligned}
{\color{jam}\prod_{\ell=k}^{d}
\frac{1}{\rpoch{\Delta_{k,\ell}-k+\ell+1}{m_\ell}}=
\left(\prod_{\ell=k}^{d-1}\frac{(\Delta_{k,\ell}-k+\ell)!}{(\Delta_{k,\ell+1}-k+\ell)!}\right)
\frac{(\Delta_{k,d}-k+d)!}{(m-k+d)!}=
\frac{\prod_{\ell=k+1}^{d}(\Delta_{k,\ell}-k+\ell)}{(m-k+d)!}},
\end{aligned}
\end{equation}
~\cref{eq:f_reduced_app} becomes
\begin{equation}
\label{eq:f_clean_app}
\begin{aligned}
\res{f}{\yd{\mu}}{\yd{m}}{\yd{\varsigma}}_\mathrm{all}(\wt{w})
=
\frac{m!}{{\color{jam}(m-k+d)!}}
{\color{defgreen}\prod_{i=1}^{k-1}
\frac{\Delta_{i,k}+k-i}{\Delta_{i,k}-i+d+m}}
{\color{jam}\prod_{\ell=k+1}^{d}(\Delta_{k,\ell}-k+\ell)}
{\color{defgreen}
\prod_{i=1}^{k-1}\prod_{\ell=k+1}^{d}
\frac{\Delta_{i,\ell}-i+\ell}{\Delta_{i,\ell}-i+\ell-1}}
.
\end{aligned}
\end{equation}
Now take \(m\to\infty\). Since \(m_d=m-\Delta_{k,d}\), Stirling gives
\begin{equation}
\begin{aligned}
\frac{m!}{(m-k+d)!}
=
m^{k-d}\left(1+\bigO{m^{-1}}\right).
\end{aligned}
\end{equation}
Also,
\begin{equation}
\begin{aligned}
\prod_{i=1}^{k-1}
\frac{1}{\Delta_{i,k}-i+d+m}
=
m^{-(k-1)}
\left(1+\bigO{m^{-1}}\right).
\end{aligned}
\end{equation}
Noting that $\multiset{d}{m}=\frac{m^{d-1}}{(d-1)!}(1+\bigO{m^{-1}})$,
\begin{equation}
\begin{aligned}
\lim_{m\to\infty} \multiset{d}{m}\,
\res{f}{\yd{\mu}}{\yd{m}}{\yd{\varsigma}}_\mathrm{all}(\wt{w})
=
\frac{1}{(d-1)!}\prod_{i=1}^{k-1}(\Delta_{i,k}+k-i)
\,
\prod_{\ell=k+1}^{d}(\Delta_{k,\ell}-k+\ell)
\,
\prod_{i=1}^{k-1}\prod_{\ell=k+1}^{d}
\frac{\Delta_{i,\ell}-i+\ell}{\Delta_{i,\ell}-i+\ell-1}
.
\end{aligned}
\end{equation}

On the other hand, we derive the dimension ratio directly from the Weyl dimension formula,
\begin{equation}
\begin{aligned}
d^{[\varsigma]}
=
\prod_{1\le i<j\le d}
\frac{\Delta_{i,j}+j-i}{j-i},
\end{aligned}
\end{equation}
We split the $i,j$ pairs into three classes.
(I) If \(1\le i<j\le k-1\), then the corresponding factor in \(d^{[\varsigma]}\) is $\frac{\Delta_{i,j}+j-i}{j-i},$
and exactly the same factor appears in \(d^{[\alpha]}\). Hence all such terms cancel.
(II) If \(k+1\le i<j\le d\), then after relabeling \(a=i-1\), \(b=j-1\), the corresponding factor in \(d^{[\alpha]}\) is again $\frac{\Delta_{i,j}+j-i}{j-i},$ so these terms also cancel completely.

Therefore, the only nontrivial contributions come from the remaining terms, namely: (i) pairs \((i,k)\) with \(1\le i\le k-1\), (ii) pairs \((k,j)\) with \(k+1\le j\le d\), (iii) pairs \((i,j)\) with \(1\le i\le k-1\) and \(k+1\le j\le d\).
For the first type, the factors appear only in \(d^{[\varsigma]}\), so they contribute $\prod_{i=1}^{k-1}\frac{\Delta_{i,k}+k-i}{k-i}.$
For the second type, these also appear only in \(d^{[\varsigma]}\), giving $\prod_{j=k+1}^{d}\frac{\Delta_{k,j}-k+j}{j-k}.$
For the third type, the contribution is
\begin{equation}
\begin{aligned}
\prod_{i=1}^{k-1}\prod_{j=k+1}^{d}
\frac{(\Delta_{i,j}+j-i)(j-i-1)}{(j-i)(\Delta_{i,j}+j-i-1)}.
\end{aligned}
\end{equation}
Putting these three surviving pieces together, we obtain
\begin{equation}
\begin{aligned}
\frac{d^{[\varsigma]}}{d^{[\alpha]}}
&=
\prod_{i=1}^{k-1}\frac{\Delta_{i,k}+k-i}{k-i}
\prod_{j=k+1}^{d}\frac{\Delta_{k,j}-k+j}{j-k}
\prod_{i=1}^{k-1}\prod_{j=k+1}^{d}
\frac{(\Delta_{i,j}+j-i)(j-i-1)}{(j-i)(\Delta_{i,j}+j-i-1)}.
\end{aligned}
\end{equation}
Finally, collecting the purely integer factors gives
\begin{equation}
\begin{aligned}
\prod_{i=1}^{k-1}(k-i)\prod_{j=k+1}^{d}(j-k)
\prod_{i=1}^{k-1}\prod_{j=k+1}^{d}\frac{j-i}{j-i-1}
=
(d-1)!,
\end{aligned}
\end{equation}
and therefore
\begin{equation}
\begin{aligned}
\frac{d^{[\varsigma]}}{d^{[\alpha]}}
=
\frac{1}{(d-1)!}
\prod_{i=1}^{k-1}(\Delta_{i,k}+k-i)
\,
\prod_{j=k+1}^{d}(\Delta_{k,j}-k+j)
\,
\prod_{i=1}^{k-1}\prod_{j=k+1}^{d}
\frac{\Delta_{i,j}-i+j}{\Delta_{i,j}-i+j-1}.
\end{aligned}
\end{equation}

Comparing the two expressions, we obtain
\begin{equation}
\begin{aligned}
\lim_{m\to\infty}c_m^{\yd{\alpha}}
=\frac{d^{[\varsigma]}}{d^{[\alpha]}}.
\end{aligned}
\end{equation}

This asymptotic coefficient is consistent with the normalization of the pulled-back POVM. Indeed, integrating~\cref{eq:pulled_back_povm} over Haar measure on pure states gives
\begin{equation}
\begin{aligned}
{\extChannel{\mathcal{T}}{\yd{\varsigma}}{\yt{1\cdots m}}{\yd{\mu}}}^\star\!\left(\int \multiset{d}{m}\psi^{\otimes m}\,\mathrm{d}\psi\right)
=
{\extChannel{\mathcal{T}}{\yd{\varsigma}}{\yt{1\cdots m}}{\yd{\mu}}}^\star\!\left(I_{\yt{1\cdots m}}\right)
=I_{\yd{\varsigma}},
\end{aligned}
\end{equation}
where we used that $\multiset{d}{m}\int \psi^{\otimes m}\,\mathrm{d}\psi=I_{\yt{1\cdots m}}$ and that the adjoint channel is unital. On the other hand, twirling $M_I$ over $U(d)$, Schur's lemma gives
\begin{equation}
\begin{aligned}
\int U^\yd{\varsigma}M_I U^{\yd{\varsigma}\dagger}\,\mathrm{d}U
=
\frac{1}{d^{[\varsigma]}}
\sum_{\yd{\beta}\,\vdash\,\left.\yd{\varsigma}\right|_{d-1}^{d}}
c_m^{\yd{\beta}} d^{[\beta]}
I_{\yd{\varsigma}}.
\end{aligned}
\end{equation}
Equating the two sides yields
\begin{equation}
\label{eq:coeff_normalization_app}
\sum_{\yd{\beta}\,\vdash\,\left.\yd{\varsigma}\right|_{d-1}^{d}}
c_m^{\yd{\beta}} d^{[\beta]}
=d^{[\varsigma]}.
\end{equation}
For the overhang-removal choice of $\yd{\mu}$ considered above, the coefficient corresponding to the diagram $\yd{\alpha}$ satisfies $\lim_{m\to\infty}
c_m^{\yd{\alpha}}d^{[\alpha]}
=d^{[\varsigma]}$. Since $c_m^{\yd{\beta}}\ge 0$ for all $\yd{\beta}$ and
$\sum_{\yd{\beta}} c_m^{\yd{\beta}} d^{[\beta]}=d^{[\varsigma]}$
for every $m$, it follows that
$\lim_{m\to\infty} c_m^{\yd{\beta}} d^{[\beta]}=0$
for every $\yd{\beta}\neq\yd{\alpha}$. Therefore, for this choice of $\yd{\mu}$, $\yd{\alpha}$ is the unique surviving $U(d-1)$ sector in the large-$m$ limit. Equivalently, the EB limit is the measurement-mediated channel with covariant POVM elements
$\{M_U
=\frac{d^{[\varsigma]}}{d^{[\alpha]}}
U^{\yd{\varsigma}}I_\yd{\alpha}U^{\yd{\varsigma}\dagger}\mid U\in U(d)\}$
followed by re-preparation of $U\ketbra{d}{d}U^\dagger$.
Using~\cref{lem:weight_permutation}, we can simplify the measurement-mediated channel further.

\begin{lemma}
\label{lem:weight_permutation}
Let $\sigma\in S_d$, and let $\mathbb{T}_{\boldsymbol\tau}^{\yd{\varsigma}}\subseteq \mathbb{W}^{\yd{\varsigma}}$ denote the type space of weight $\boldsymbol\tau$. Then
$
U_\sigma^{\yd{\varsigma}}\,\mathbb{T}_{\boldsymbol\tau}^{\yd{\varsigma}}
=
\mathbb{T}_{\sigma(\boldsymbol\tau)}^{\yd{\varsigma}}$, where we recall that
$\sigma(\boldsymbol\tau)_i=\tau_{\sigma^{-1}(i)}$.
In particular, if both $\mathbb{T}_{\boldsymbol\tau}^{\yd{\varsigma}}$ and $\mathbb{T}_{\sigma(\boldsymbol\tau)}^{\yd{\varsigma}}$ are one dimensional, then for any nonzero vector $\ket{v}\in \mathbb{T}_{\boldsymbol\tau}^{\yd{\varsigma}}$,
$U_\sigma^{\yd{\varsigma}}\ket{v}=c\,\ket{v^\prime}$
for some nonzero scalar $c$, where $\ket{v^\prime}$ is any nonzero vector in $\mathbb{T}_{\sigma(\boldsymbol\tau)}^{\yd{\varsigma}}$.
\end{lemma}

\begin{proof}
Let $H=\sum_i h_i\ket{i}\!\bra{i}$ be any Cartan element in the fundamental representation. Since $U_\sigma\ket{i}=\ket{\sigma(i)}$, we have
\begin{equation}
U_\sigma^\dagger H U_\sigma
=
\sum_i h_i \ket{\sigma^{-1}(i)}\!\bra{\sigma^{-1}(i)}
=
\sum_i h_{\sigma(i)}\ket{i}\!\bra{i}.
\end{equation}
Passing to the irrep $\yd{\varsigma}$ gives
\begin{equation}
\bigl(U_\sigma^{\yd{\varsigma}}\bigr)^\dagger H^{\yd{\varsigma}} U_\sigma^{\yd{\varsigma}}
=
\left(\sum_i h_{\sigma(i)}\ket{i}\!\bra{i}\right)^{\yd{\varsigma}}.
\end{equation}
Now let $\ket{v}\in \mathbb{T}_{\boldsymbol\tau}^{\yd{\varsigma}}$. By definition of the type space, $H^{\yd{\varsigma}}\ket{v}
=
\left(\sum_i \tau_i h_i\right)\ket{v}.$
Hence
\begin{equation}
\begin{aligned}
H^{\yd{\varsigma}} U_\sigma^{\yd{\varsigma}}\ket{v}
&=
U_\sigma^{\yd{\varsigma}}
\left(\sum_i h_{\sigma(i)}\ket{i}\!\bra{i}\right)^{\yd{\varsigma}}\ket{v}\\
&=
\left(\sum_i \tau_i h_{\sigma(i)}\right) U_\sigma^{\yd{\varsigma}}\ket{v}\\
&=
\left(\sum_i \tau_{\sigma^{-1}(i)} h_i\right) U_\sigma^{\yd{\varsigma}}\ket{v}.
\end{aligned}
\end{equation}
Therefore $U_\sigma^{\yd{\varsigma}}\ket{v}$ has weight $\sigma(\boldsymbol\tau)$, so $U_\sigma^{\yd{\varsigma}}\,\mathbb{T}_{\boldsymbol\tau}^{\yd{\varsigma}}
\subseteq
\mathbb{T}_{\sigma(\boldsymbol\tau)}^{\yd{\varsigma}}.$
Applying the same argument to $\sigma^{-1}$ gives the reverse inclusion, hence
$U_\sigma^{\yd{\varsigma}}\,\mathbb{T}_{\boldsymbol\tau}^{\yd{\varsigma}}
=
\mathbb{T}_{\sigma(\boldsymbol\tau)}^{\yd{\varsigma}}.$
If both type spaces are one dimensional, then the image of any nonzero vector in $\mathbb{T}_{\boldsymbol\tau}^{\yd{\varsigma}}$ must be a nonzero scalar multiple of any nonzero vector in $\mathbb{T}_{\sigma(\boldsymbol\tau)}^{\yd{\varsigma}}$, proving the final statement.
\end{proof}

For each \(k\), consider the Kostka number \(K^{\yd{\varsigma}}_{\boldsymbol{\tau}}\) with
$\boldsymbol{\tau}
=
[\varsigma_1,\ldots,\varsigma_{k-1},\varsigma_{k+1},\ldots,\varsigma_d,\varsigma_k].$
Since \(\boldsymbol{\tau}\) is a permutation of the highest weight \(\yd{\varsigma}\), by permutation invariance of Kostka numbers we have \(K^{\yd{\varsigma}}_{\boldsymbol{\tau}}=1\). Thus there is a unique Weyl tableau (WT) \(\wt{w}\vdash\yd{\varsigma}\) with \(\boldsymbol{\#}(\wt{w})=\boldsymbol{\tau}\). The second row of its GT pattern is \(\yd{\alpha}\). By~\cref{lem:weight_permutation}, if we apply the permutation \(\sigma\) defined in Ref.~\cite[Eq.~(S56)]{QPA26S}, which sends $\yd{\varsigma}$ to $\boldsymbol{\tau}$ and \(k\) to \(d\), then \(U^\yd{\varsigma}_\sigma\) maps the unique GT basis state \(\ket{\wt{\mathrm {lw}}}\) to \(\ket{\wt{w}}\). Therefore, since the final POVM is unitarily invariant, we may equivalently describe it as $\{M_U
=d^\yd{\varsigma}
U^{\yd{\varsigma}}\ketbra{\wt{\mathrm{lw}}}{\wt{\mathrm{lw}}}U^{\yd{\varsigma}\dagger}\mid U\in U(d)\}$, and then outputting the correspondingly rotated basis state \(U\ketbra{k}{k}U^\dagger\).

\subsection{Optimal EB Risk}
\label{subsec:optimal_eb_risk_supp}

We characterize the risk that the EB protocol described in the previous section can achieve using~\cref{cor:limit_of_optimal_symmetric} by studying the $m\to\infty$ limit of Ref.~\cite[Theorem~S5.32]{QPA26S}. We set
\begin{equation}
\begin{aligned}
{}^k\widetilde{f}^{\yd{\varsigma}}
\coloneqq
\lim_{m\to\infty}
\res{{^kf}}{\yd{\varsigma}}{\yd{m}}{\yd{\lambda}}_{\mathrm{one}},
\end{aligned}
\end{equation}
whenever the limit exists. In particular, the sector-wise lower bound becomes
\begin{equation}
\begin{aligned}
{}^k\widetilde{f}^{\yd{\varsigma}}
\geq
1
-\frac{1}{\Delta_{k-1,k}+1}\sum_{i=1}^{k-1}\frac{p_i}{D_{i,k}}
-\frac{1}{\Delta_{k,k+1}}\sum_{i=k+1}^{d}\frac{p_i}{D_{k,i}}
-\frac{d-k}{\Delta_{k,k+1}+(d-k)} \, .
\end{aligned}
\end{equation}
Likewise, the upper bound becomes
\begin{equation}
\begin{aligned}
{}^k\widetilde{f}^{\yd{\varsigma}}
\leq
1
-\frac{d-k}{\Delta_{k,d}+(d-k)}
\left(
1-\frac{1}{\Delta_{k,k+1}}
\left(
1+\sum_{i=k+1}^{d}\frac{p_i}{D_{k,i}}
\right)
\right)
\, .
\end{aligned}
\end{equation}

\begin{theorem}[Lower bound for the EB one-site risk]
\label{thm:one_site_risk_lower_bound}
Let $0<k<d$. Assume
\begin{equation}
n\ge \frac{4}{D_{k,\mathrm{min}}}\left(1+\frac{1}{D_{k,\mathrm{min}}}\right).
\end{equation}
Then the optimal EB risk is lower bounded by
\begin{equation}
{}^k\widetilde{\mathcal L}_{(\cdot),\mathrm{one}}^\ast
\ge
\left(
1-8e^{-\frac{nD_{k,\mathrm{min}}^2}{256}}
\right)
\frac{d-k}{2(n+d-k)}.
\end{equation}
\end{theorem}

\begin{proof}
Using $\Delta_{k,d}\leq n$, the sector-wise upper bound gives
\begin{equation}
\label{eq:loss_after_id}
{}^k\widetilde{f}^{\yd{\varsigma}}
\leq
1
-\frac{d-k}{n+(d-k)}
\left(
1-\frac{1}{\Delta_{k,k+1}}
\left(
1+\sum_{i=k+1}^{d}\frac{p_i}{D_{k,i}}
\right)
\right)
\, .
\end{equation}
Now choose $\alpha\coloneqq D_{k,\mathrm{min}}/2$ for Ref.~\cite[Lemma~S2.5]{QPA26S}. For the typical sectors we have
\begin{equation}
\label{eq:pr_deviation_normalized_row_length}
\Pr\!\left(\,|\overline{\Delta}_{k,k+1}-D_{k,k+1}|\ge \frac{D_{k,\mathrm{min}}}{2}\,\right)
\le
8e^{-\frac{nD_{k,\mathrm{min}}^2}{256}} \, ,
\end{equation}
and also
\begin{equation}
\overline{\Delta}_{k,k+1}\ge D_{k,k+1}-\frac{D_{k,\mathrm{min}}}{2}\geq\frac{D_{k,\mathrm{min}}}{2}.
\end{equation}
Therefore we have on the typical sectors
\begin{equation}
{}^k\widetilde{f}^{\yd{\varsigma}}
\le 1 - \frac{d-k}{n+(d-k)}\left(1-\frac{2}{nD_{k,\mathrm{min}}}\left(1+\frac{1}{D_{k,\mathrm{min}}}\right)\right).
\end{equation}
By the assumption on $n$,
\begin{equation}
\frac{2}{nD_{k,\mathrm{min}}}\left(1+\frac{1}{D_{k,\mathrm{min}}}\right)\le \frac12.
\end{equation}
Altogether this means
\begin{equation}
{}^k\widetilde{f}^{\yd{\varsigma}}
\le 1 - \frac{d-k}{2(n+d-k)} \, .
\end{equation}
Together with~\cref{eq:pr_deviation_normalized_row_length} we get the result.
\end{proof}

\begin{theorem}[Asymptotic EB fidelity for eigenstate tomography]
\label{thm:asymptotic_eb_fidelity_eigenstate_tomography}
For a non-degenerate sorted spectrum \(\boldsymbol{p}\), the one-site EB utility for the \(k\)-th eigenstate tomography task satisfies
\begin{equation}
{}^k\widetilde{\mathcal F}_{(\cdot),\mathrm{one}}^\ast
=
1-\frac{1}{n}
\left(
\sum_{i\neq k}\frac{p_i}{D_{k,i}^2}
+
\sum_{i=k+1}^d\frac{1}{D_{k,i}}
\right)
+\lilo[d,\boldsymbol{p}]{n^{-1}}.
\end{equation}
\end{theorem}

\begin{proof}
Ref.~\cite[Theorems~S4.29 and~S6.2]{QPA26S} gives the sector-wise asymptotically optimal one-site expansion for symmetric QPA protocols, uniformly over the typical region \(\mathsf{R}_\epsilon\) defined there. Applying the EB limit in~\cref{thm:EB_risk} identifies the optimal EB risk with the \(m\to\infty\) limit of the corresponding optimal symmetric one-site risk, up to a vanishing error. Because the limit \(m\to\infty\) is taken first, the dimension dependence in the EB approximation does not affect the displayed coefficient. Taking this limit in the one-site expansion removes the \(1/m\) term. The Schur--Weyl averaging step is the same as in Ref.~\cite[Lemma~S2.3]{QPA26S}, giving the displayed expansion. Optimality similarly follows from the sector-wise asymptotic optimality statement and its overall consequence in Ref.~\cite[Theorems~S6.2 and~S6.3]{QPA26S}.
\end{proof}

\begin{corollary}[Sample complexity of eigenstate tomography]
\label{lem:eb_sample_complexity_lower}
Write \({^k}\widetilde{\mathcal L}_{\mathrm{one}}^\ast(n)\) as the optimal EB one-site loss for $n$ inputs and let $\widetilde n(\varepsilon)\coloneqq \min\{n\mid {^k}\widetilde{\mathcal L}_{\mathrm{one}}^\ast(n)\le \varepsilon\}$. Assume the conditions of~\cref{thm:one_site_risk_lower_bound}. Then, as $\varepsilon\to0$,
\begin{equation}
\widetilde n(\varepsilon)\ge\frac{d-k}{2\varepsilon}\mleft(1-\lilo[d,k][\varepsilon]{1}\mright).
\end{equation}
Here the sufficiently small-$\varepsilon$ threshold may depend on $d$ and $k$.
\end{corollary}

\begin{proof}
Set $a\coloneqq d-k$. We first establish $\widetilde n(\varepsilon)=\lilom[d,k][\varepsilon]{1}$. For every finite $N$, the optimal QPA risk ${^k}\widetilde{\mathcal L}_{\mathrm{one}}^\ast(N)$ is strictly positive, since finitely many copies cannot identify an unknown eigenstate exactly over the continuous unitary orbit. Moreover, the optimal risk is nonincreasing in the input size, so for every $n\leq N$,
\begin{equation*}
{^k}\widetilde{\mathcal L}_{\mathrm{one}}^\ast(n)
\geq
{^k}\widetilde{\mathcal L}_{\mathrm{one}}^\ast(N)
>0.
\end{equation*}
Consequently, once $\varepsilon<{^k}\widetilde{\mathcal L}_{\mathrm{one}}^\ast(N)$, no $n\leq N$ can achieve risk at most $\varepsilon$, and hence $\widetilde n(\varepsilon)>N$. Since this holds for every finite $N$, it is precisely the pointwise statement $\widetilde n(\varepsilon)=\lilom[d,k][\varepsilon]{1}$. Because $D_{k,\mathrm{min}}>0$ is fixed, this also gives $8e^{-\frac{\widetilde n(\varepsilon)D_{k,\mathrm{min}}^2}{256}}=\lilo[d,k][\varepsilon]{1}$. For all sufficiently small $\varepsilon$, $\widetilde n(\varepsilon)$ therefore satisfies the sample-size condition of~\cref{thm:one_site_risk_lower_bound}, and hence
\begin{equation}
\varepsilon
\geq
{^k}\widetilde{\mathcal L}_{\mathrm{one}}^\ast\mleft(\widetilde n(\varepsilon)\mright)
\geq
\mleft(1-\lilo[d,k][\varepsilon]{1}\mright)\frac{a}{2(\widetilde n(\varepsilon)+a)}.
\end{equation}
Rearranging and absorbing the additive term $-a$ into the $\lilo[d,k][\varepsilon]{1}$ remainder gives the claim.
\end{proof}

\section{DME}
\label{sec:dme_supp}

DME can be similarly formulated as a channel-level CQI task. Given $n$ copies of an unknown mixed state $\rho\in\mathcal{B}(\mathbb{C}^d)$ and a simulation time $T>0$, the target object is the unitary-evolution channel
\begin{equation}
{^T}\!\!\mathcal A_\rho(\sigma)\coloneqq e^{-i\rho T}\sigma e^{i\rho T},
\label{eq:dme_state_level_target}
\end{equation}
acting on an auxiliary input register. Equivalently, the state-level DME input family is $\mathcal X\coloneqq\{\Sigma=\rho^{\otimes n}\otimes\sigma_{\mathrm{data}}\mathrel{\mid}\rho,\sigma_{\mathrm{data}}\in\mathcal B(\mathbb C^d)\text{ are density matrices}\}$, with the copies of $\rho$ in the program register and $\sigma_{\mathrm{data}}$ in the data register.
To define the worst-case risk, we optimize over all inputs,
\begin{equation}
    \sup_{\Sigma\in\mathcal{X}}
    L\mleft(\{
    {^T}\!\!\mathcal A_\rho
  (\sigma_{\mathrm{data}}(\Sigma))\},\mathcal{T}(\Sigma)
    \mright).
\end{equation}
Here $\sigma_{\mathrm{data}}(\Sigma)\coloneqq\Tr_{\mathrm{prog}_{1,\ldots,n}}(\Sigma)$ is the data-register marginal of $\Sigma$.
For the trace-distance loss, the corresponding channel-level error is the induced trace-norm error, which is no larger than the diamond-norm error,
\begin{equation}
    \sup_{\rho}
    \left\|
    {^T}\!\!\mathcal A_\rho
    -
    \mathcal{T}(\rho^{\otimes n}\otimes \cdot)\right\|_{1\to1}
    \leq
    \sup_{\rho}
    \left\|
    {^T}\!\!\mathcal A_\rho
    -
    \mathcal{T}(\rho^{\otimes n}\otimes \cdot)\right\|_{\diamond},
\end{equation}
thereby allowing us to also cast DME as a channel-level approximation problem.

The original DME construction of Refs.~\cite{LMR14S,GKPP+25S} realizes this target by sequential infinitesimal interactions with the copies of $\rho$ and achieves diamond-norm error $\varepsilon$ with sample complexity $n=\bigO{\frac{T^2}{\varepsilon}},$
independent of the Hilbert-space dimension.

For partially coherent protocols that are entanglement breaking on the program register, classical information about $\rho$ must first be extracted and then used to synthesize an approximation to ${^T}\!\!\mathcal A_\rho$. When we consider the corresponding partially coherent protocol, only the $\rho^{\otimes n}$ register is measured, while the coherence in $\sigma_{\mathrm{data}}$ remains available. Below, we show that, for sufficiently small $\varepsilon$, with a threshold depending only on $T$, this necessarily requires at least \(n=\bigOm[][\varepsilon]{\sin^2\left(T/2\right)d/\varepsilon}\) input copies.
\begin{lemma}[Partially coherent DME as a classically programmed channel]
Any partially coherent strategy in our setting for DME can be written as a classically programmed channel. Namely, there exists a POVM \(\{M_x\}_x\) on the \(n\) input copies and, for each outcome \(x\), a channel \(\mathcal N_x\) acting on the data register such that, for every input \(\sigma_{\mathrm{data}}\),
\begin{equation}
\mathcal T
\left(
\rho_{\mathrm{in}}^{\otimes n}\otimes\sigma_{\mathrm{data}}
\right)
=
\sum_x
\Tr\left(M_x\rho_{\mathrm{in}}^{\otimes n}\right)
\mathcal N_x
(\sigma_{\mathrm{data}}) .
\end{equation}
\end{lemma}

\begin{proof}
Since the protocol is entanglement breaking on the program register, the input copies first pass through an EB channel before the remaining quantum operation is applied to the data register. By the structure theorem for EB channels, this channel can be written, without loss of generality, in measurement-mediated form: $\sum_x
\Tr\left(
M_x\rho^{\otimes n}
\right)x,$
for some POVM \(\{M_x\}_x\), where \(x=\ketbra{x}{x}\) is stored in a classical register \(X\). The remaining part of the strategy receives this classical register together with the data register. The action of the protocol is then \(\mathcal{T}(\rho_{\mathrm{in}}^{\otimes n}\otimes\sigma_{\mathrm{data}})=\sum_x\Tr\left(M_x\rho_{\mathrm{in}}^{\otimes n}\right)\mathcal{N}(\sigma_{\mathrm{data}}\otimes x)\). Since \(X\) is classical, we define the classically controlled channel $\mathcal N_x(\cdot)=\mathcal{N}(\cdot\otimes x)$. Applying this to the data state $\sigma_{\mathrm{data}}$ gives the claimed form.
\end{proof}

Since we are proving a lower bound, we restrict attention to the following candidate family. Let $d'=d-1$. For $\boldsymbol{\theta}\dot{=}(\theta_1,\ldots,\theta_{d'})\in\mathbb R^{d'}$ with
$\theta\coloneqq\|\boldsymbol{\theta}\|\leq1$, define
\begin{equation}
\ket{\vartheta}
=
\sqrt{1-\theta^2}\,\ket{0}
+
\sum_{j=1}^{d'}\theta_j\ket{j},
\end{equation}
with $\vartheta$ the corresponding density matrix. Evaluate the channel on the fixed input
$\ketbra{0}{0}$. Define the corresponding ideal output state $\ket{\gamma_{\boldsymbol{\theta}}}
=
e^{-iT\vartheta}\ket{0}$
and $\gamma_{\boldsymbol{\theta}}$ for the corresponding pure-state density matrix.
Then ${^T}\!\!\mathcal A_\vartheta(\ketbra{0}{0})=\gamma_{\boldsymbol{\theta}}$.
Since $\vartheta$ is a rank-one projector,
$e^{-iT\vartheta}=I+(e^{-iT}-1)\vartheta$. Therefore
\begin{equation}
\ket{\gamma_{\boldsymbol{\theta}}}
=
A(\boldsymbol{\theta})\ket{0}
+
B(\boldsymbol{\theta})\sum_{j=1}^{d'}\theta_j \ket{j} ,
\end{equation}
where
\begin{equation}
A(\boldsymbol{\theta})
=
1+(e^{-iT}-1)(1-\theta^2),
\qquad
B(\boldsymbol{\theta})
=
(e^{-iT}-1)\sqrt{1-\theta^2}.
\end{equation}

\begin{lemma}[Embedding of the output]
\label{lem:embedding_output}
There exist constants \(r_0>0\) and \(a_T>0\), depending only on \(T\), such that for all \(\boldsymbol{\theta},\boldsymbol{\eta}\in\mathbb R^{d'}\) with \(\|\boldsymbol{\theta}\|,\|\boldsymbol{\eta}\|\le r_0\),
\begin{equation}
1-F(\gamma_{\boldsymbol{\theta}},\gamma_{\boldsymbol{\eta}}) \ge a_T\|\boldsymbol{\theta}-\boldsymbol{\eta}\|^2.
\end{equation}
Moreover, \(r_0\) can be chosen so that
\(a_T=\frac{1}{4}\sin^2\left(\frac{T}{2}\right)\).
\end{lemma}
\begin{proof}
When the \(\ket{0}\) coefficient is nonzero, we can reparametrize
\(\ket{\gamma_{\boldsymbol{\theta}}}\) by the affine coordinates
\begin{equation}
\mathbf G(\boldsymbol{\theta})
=
g(\boldsymbol{\theta})\boldsymbol{\theta}
=
\frac{B(\boldsymbol{\theta})}{A(\boldsymbol{\theta})}
\boldsymbol{\theta}.
\end{equation}

Specifically, \(\mathbf G(\boldsymbol 0)=\boldsymbol 0\). Moreover, \(\mathbf G\) is continuously differentiable near \(\boldsymbol 0\) with $\boldsymbol{\nabla}\mathbf G(\boldsymbol 0)
=
(1-e^{iT})\mathbf I_{d'}.$
Choose \(r_0>0\), depending only on \(T\), such that for every
\(\boldsymbol{\zeta}\) with \(\|\boldsymbol{\zeta}\|\le r_0\),
\begin{equation}
A(\boldsymbol{\zeta})\neq0,
\qquad
\left\|\mathbf G(\boldsymbol{\zeta})\right\|\le \frac12,
\qquad
\mathbf E(\boldsymbol{\zeta})
\coloneqq
\boldsymbol{\nabla}\mathbf G(\boldsymbol{\zeta})
-(1-e^{iT})\mathbf I_{d'},
\qquad
\left\|\mathbf E(\boldsymbol{\zeta})\right\|_\infty
\le
\frac12\left|1-e^{iT}\right|
.
\end{equation}
Moreover, up to a global phase, the target state reads
\begin{equation}
\ket{\gamma_{\boldsymbol{\zeta}}}
\sim
\frac{
\ket{0}
+
\sum_{j=1}^{d'} G_j(\boldsymbol{\zeta})\ket{j}
}{
\sqrt{1+\|\mathbf G(\boldsymbol{\zeta})\|^2}
}.
\end{equation}
Then, for \(\boldsymbol{\theta},\boldsymbol{\eta}\) in the ball,
\begin{align}
\mathbf G(\boldsymbol{\theta})-\mathbf G(\boldsymbol{\eta})
&=
\int_0^1
\boldsymbol{\nabla}\mathbf G\!\left(
\boldsymbol{\eta}
+\lambda(\boldsymbol{\theta}-\boldsymbol{\eta})
\right)\cdot
(\boldsymbol{\theta}-\boldsymbol{\eta})\,\mathrm{d}\lambda\nonumber\\
&=
g(\boldsymbol{0})(\boldsymbol{\theta}-\boldsymbol{\eta})
+
\int_0^1
\mathbf E\mleft(\boldsymbol\eta+\lambda(\boldsymbol{\theta}-\boldsymbol{\eta})\mright)\cdot(\boldsymbol{\theta}-\boldsymbol{\eta})\,\mathrm{d}\lambda.
\end{align}
Hence
\begin{equation}
\left\|
\mathbf G(\boldsymbol{\theta})-\mathbf G(\boldsymbol{\eta})
\right\|^2
\ge
\sin^2\left(\frac{T}{2}\right)
\|\boldsymbol{\theta}-\boldsymbol{\eta}\|^2 .
\end{equation}
Since
\(\|\mathbf G(\boldsymbol{\theta})\|,\|\mathbf G(\boldsymbol{\eta})\|\le1/2\),
a direct computation using the affine parametrization gives
\begin{equation}
1-F(\gamma_{\boldsymbol{\theta}},\gamma_{\boldsymbol{\eta}})
=
1-
\left|
\left\langle
\frac{\ket{0}+\sum_{j=1}^{d'} G_j(\boldsymbol{\theta})\ket{j}}
{\sqrt{1+\|\mathbf G(\boldsymbol{\theta})\|^2}},
\frac{\ket{0}+\sum_{j=1}^{d'} G_j(\boldsymbol{\eta})\ket{j}}
{\sqrt{1+\|\mathbf G(\boldsymbol{\eta})\|^2}}
\right\rangle
\right|^2
\ge
\frac14
\left\|
\mathbf G(\boldsymbol{\theta})
-
\mathbf G(\boldsymbol{\eta})
\right\|^2 .
\end{equation}
Altogether, we obtain
\begin{equation}
1-F(\gamma_{\boldsymbol{\theta}},\gamma_{\boldsymbol{\eta}})
\ge
\frac{1}{4}\sin^2\left(\frac{T}{2}\right)
\|\boldsymbol{\theta}-\boldsymbol{\eta}\|^2 .
\end{equation}
Therefore the lemma holds with
\(a_T=\frac{1}{4}\sin^2\left(\frac{T}{2}\right)\).
\end{proof}

\begin{theorem}[Sample-complexity lower bound on worst-case partially coherent DME]
\label{thm:dme_partially_coherent_sample_lower}
Assume $d\geq 2$ and simulation time $T\not\equiv 0 \pmod {2\pi}$. There exist a universal constant $c>0$ and a constant $\varepsilon_T>0$, depending only on $T$, such that the following holds.
Suppose that a partially coherent DME protocol achieves error \(\varepsilon\) in the induced trace norm uniformly over all density matrices $\rho$, i.e.,
\begin{equation}
\sup_\rho
\left\|\mathcal T(\rho^{\otimes n}\otimes\cdot)-{^T}\!\!\mathcal A_\rho\right\|_{1\to1}
\le
\varepsilon.
\end{equation}
Then, for $0<\varepsilon<\varepsilon_T$,
$n
\ge
c\sin^2\!\left(\frac{T}{2}\right)\frac{d}{\varepsilon}.$
\end{theorem}
\begin{proof}
We first restrict the program state to the candidate subfamily defined above with $\|\boldsymbol{\theta}\|\leq r_0$ and evaluate the DME protocol on a test input, here chosen to be $\ketbra{0}{0}$. A good DME protocol then induces a good state-transformation channel from the candidate program state $\vartheta^{\otimes n}$ to the target state $\gamma_{\boldsymbol{\theta}}$.
Denote the output state by
\begin{equation}
\sigma^\prime_{\boldsymbol{\theta}}
=
\mathcal{T}(\vartheta^{\otimes n}\otimes\ketbra{0}{0}).
\end{equation}
Moreover, for each transcript $x$, define the conditional output
$\sigma_x=\mathcal N_x(\ketbra{0}{0})$, so $\sigma^\prime_{\boldsymbol{\theta}}=\left\langle \sigma_X\right\rangle_{\boldsymbol{\theta}}$.
Since the induced trace norm upper bounds the trace distance on every $\sigma$,
\begin{equation}
\langle1-F(\gamma_{\boldsymbol{\theta}},\sigma_X)\rangle_{\boldsymbol{\theta}}= 1-F(\gamma_{\boldsymbol{\theta}},\sigma^\prime_{\boldsymbol{\theta}})\le D_{\mathrm{tr}}(\sigma^\prime_{\boldsymbol{\theta}},\gamma_{\boldsymbol{\theta}})
\le
\left\|\mathcal T(\vartheta^{\otimes n}\otimes\cdot)-{^T}\!\!\mathcal A_\vartheta\right\|_{1\to1}
\le
\varepsilon.
\end{equation}
The first equality uses that $\gamma_{\boldsymbol{\theta}}$ is pure, so the fidelity is linear in its second argument.

We show that this scalar loss bound yields an accurate estimator for the
parameter \(\boldsymbol{\theta}\). We first perform the POVM on the input copies
and record the classical outcome \(x\).
Conditioned on \(X=x\), we estimate \(\boldsymbol{\theta}\) by projecting \(\sigma_x\) onto the target family \(\{\gamma_{\boldsymbol{\theta}'}\}\) in trace distance:
\begin{equation}
\widehat{\boldsymbol{\theta}}(x)
\in
\operatorname*{argmin}_{\|\boldsymbol{\theta}'\|\le r_0}
D_{\mathrm{tr}}(\sigma_x,\gamma_{\boldsymbol{\theta}'}).
\end{equation}
By the minimizing property of $\widehat{\boldsymbol\theta}(x)$ and the triangle inequality, with the true state $\gamma_{\boldsymbol\theta}$ as an admissible point in the target family, we have
\begin{equation}
D_{\mathrm{tr}}(\gamma_{\boldsymbol{\theta}},
\gamma_{\widehat{\boldsymbol{\theta}}(x)})
\le
2D_{\mathrm{tr}}(\sigma_x,\gamma_{\boldsymbol{\theta}}).
\end{equation}
Therefore, using the pure-state identity and the Fuchs--van de Graaf inequality,
\begin{equation}
1-F(\gamma_{\boldsymbol{\theta}},\gamma_{\widehat{\boldsymbol{\theta}}(x)})
=
D_{\mathrm{tr}}(\gamma_{\boldsymbol{\theta}},
\gamma_{\widehat{\boldsymbol{\theta}}(x)})^2
\le
4D_{\mathrm{tr}}(\sigma_x,\gamma_{\boldsymbol{\theta}})^2
\le
4(1-F(\gamma_{\boldsymbol{\theta}},\sigma_x)).
\end{equation}

Applying~\cref{lem:embedding_output} gives
$a_T\|\boldsymbol{\theta}-\widehat{\boldsymbol{\theta}}(x)\|^2\le4\left(1-F(\gamma_{\boldsymbol{\theta}},\sigma_x)\right)$. Taking expectation over $X$ gives the mean square error bound
\begin{equation}
\left\langle
\|\boldsymbol{\theta}-\widehat{\boldsymbol{\theta}}(X)\|^2
\right\rangle_{\boldsymbol{\theta}}
\le
\frac{4}{a_T}
\left\langle
1-F(\gamma_{\boldsymbol{\theta}},\sigma_X)
\right\rangle_{\boldsymbol{\theta}}
\le
\frac{4\varepsilon}{a_T}.
\end{equation}
Thus a good partially coherent DME protocol yields an estimator of the parameter $\boldsymbol{\theta}$ whose mean-squared error is at most
$4\varepsilon/a_T$.

We prove that such estimation will necessarily incur a dimension-dependent sample complexity. We further reduce the problem to a finite input state family $\mathcal{X}$ by embedding a hypercube in the local parameter space. Let $\tau$ be a sign vector
drawn uniformly from $\{+1,-1\}^{d'}$, and define $\boldsymbol{\theta}_{\boldsymbol{\tau}}=\alpha\boldsymbol{\tau}$.
Choose
\begin{equation}
\alpha^2
=
\min\left\{
\frac{1}{64n},
\frac{r_0^2}{d'}
\right\}.
\end{equation}
Then $d'\alpha^2\le r_0^2$, so the hypercube lies in the admissible range of $\mathbf{\theta}$.

The shrinking hypercube forces the input states to be close. Let $\boldsymbol{\tau}^{(j)}$ be obtained from $\boldsymbol{\tau}$ by flipping the $j$-th coordinate. Then $\braket{\vartheta_{\boldsymbol{\tau}}}{\vartheta_{\boldsymbol{\tau}^{(j)}}}=1-2\alpha^2$.
For $n$ copies,
\begin{equation}
\left|\braket{\vartheta_{\boldsymbol{\tau}}^{\otimes n}}{\vartheta_{\boldsymbol{\tau}^{(j)}}^{\otimes n}}\right|^2
=
(1-2\alpha^2)^{2n}.
\end{equation}
Since $n\alpha^2\le 1/64$, and using
$\ln(1-y)\ge -2y$ for $0\le y\le 1/2$, with $y=2\alpha^2$, we obtain
\begin{equation}
(1-2\alpha^2)^{2n}
\ge
e^{-8n\alpha^2}
\ge
e^{-1/8}.
\end{equation}
Therefore the normalized trace distance between the neighboring $n$-copy
states is bounded by
\begin{equation}
D_{\mathrm{tr}}(\vartheta_\tau^{\otimes n},
\vartheta_{\tau^{(j)}}^{\otimes n})
\le
\sqrt{1-e^{-1/8}}
\leq \frac{1}{2} .
\end{equation}
By the Helstrom bound for binary state discrimination, when both $\vartheta_{\boldsymbol\tau}$ and $\vartheta_{\boldsymbol\tau^{(j)}}$ occur with probability $\frac12$, any discriminator $\widehat\tau_j$ for their differing $j$th coordinate has average error probability at least
\begin{equation}
\Pr(\widehat\tau_j\neq \tau_j)\geq\frac12
\left(
1-
D_{\mathrm{tr}}(\vartheta_\tau^{\otimes n},
\vartheta_{\tau^{(j)}}^{\otimes n})
\right)
\ge
\frac14 .
\end{equation}

Now take any estimator $\widehat{\boldsymbol{\theta}}(X)$, and define the induced discriminator $\widehat\tau_j(X)=\mathrm{sign}(\widehat\theta_j(X))$, with $\mathrm{sign}(0)$ chosen arbitrarily. If $\widehat\tau_j(X)\neq \tau_j$, then $(\widehat\theta_j(X)-\theta_{\boldsymbol{\tau},j})^2\ge \alpha^2$. Since $\widehat\tau_j(X)$ can be used as a state discriminator,
\begin{equation}
\left\langle
\|\widehat{\boldsymbol{\theta}}(X)-\boldsymbol{\theta}_{\boldsymbol{\tau}}\|^2
\right\rangle_{\boldsymbol{\tau},X}
\ge
\alpha^2
\sum_{j=1}^{d'}
\Pr(\widehat\tau_j(X)\neq \tau_j)
\ge
\frac{1}{4}d'\alpha^2
\ge
\frac{1}{256}
\min\left\{
\frac{d-1}{n},
r_0^2
\right\}.
\end{equation}

Finally, we combine this lower bound with the estimator constructed from the partially coherent DME protocol.
From the DME guarantee, averaged over the hypercube $\mathcal{X}$,
\begin{equation}
\frac{4\varepsilon}{a_T}\ge\left\langle
\|\widehat{\boldsymbol{\theta}}(X)-\boldsymbol{\theta}_{\boldsymbol{\tau}}\|^2
\right\rangle_{\boldsymbol{\tau},X}\ge\frac{1}{256}
\min\left\{
\frac{d-1}{n},
r_0^2
\right\}.
\end{equation}

If $0<\varepsilon<\varepsilon_T=\frac{a_T}{1024}r_0^2$, the condition simplifies to $n\ge\frac{a_T}{1024}\frac{d-1}{\varepsilon}$. For $d\ge 2$, using $a_T=\frac14\sin^2(T/2)$, this implies
\begin{equation}
n
\geq
\frac{1}{8192}
\sin^2\!\left(\frac{T}{2}\right)
\frac{d}{\varepsilon}.
\end{equation}
Thus the theorem holds with the universal constant $c=1/8192$.
\end{proof}

\putbib[ci_si.bib]
\end{bibunit}

\end{document}